\documentclass[12pt]{article}
\pdfoutput=1
\usepackage{amsmath}
\usepackage{graphicx,psfrag,epsf}
\usepackage{enumerate}
\usepackage{natbib}
\usepackage{url} 

\newcommand{\blind}{0}

\usepackage{algorithm}
\usepackage{algpseudocode}
\usepackage{dsfont}
\usepackage{amsthm}
\usepackage{xr}
\newcommand{\beq}{\begin{equation}}
\newcommand{\eeq}{\end{equation}}
\newcommand{\bmone}{\mathbf{1}}
\newcommand{\bma}{\mathbf{a}}
\newcommand{\bmb}{\mathbf{b}}
\newcommand{\bmp}{\mathbf{p}}
\newcommand{\bmr}{\mathbf{r}}
\newcommand{\bmx}{\mathbf{x}}
\newcommand{\bmy}{\mathbf{y}}
\newcommand{\bmz}{\mathbf{z}}
\newcommand{\bmA}{\mathbf{A}}
\newcommand{\bmH}{\mathbf{H}}
\newcommand{\bmI}{\mathbf{I}}
\newcommand{\bmX}{\mathbf{X}}
\newcommand{\bmbeta}{\boldsymbol{\beta}}
\newcommand{\bmgamma}{\boldsymbol{\gamma}}
\newcommand{\bmepsilon}{\boldsymbol{\epsilon}}
\newcommand{\bmlambda}{\boldsymbol{\lambda}}
\newcommand{\bmphi}{\boldsymbol{\phi}}
\newcommand{\bmPhi}{\boldsymbol{\Phi}}

\DeclareMathOperator*{\argmin}{arg\,min}

\newtheorem{theorem}{\bf Theorem}[section]
\newtheorem{proposition}[theorem]{Proposition}

\newtheorem{example}[theorem]{Example}
\newtheorem{definition}[theorem]{Definition}

\begin{document}

\def\spacingset#1{\renewcommand{\baselinestretch}%
{#1}\small\normalsize} \spacingset{1}


\if0\blind
{
  \title{\bf Minimizing Sum of Truncated Convex Functions and Its Applications}
  \author{Tzu-Ying Liu and Hui Jiang\thanks{Please send all correspondence to jianghui@umich.edu.}\\
    Department of Biostatistics, University of Michigan\\
    Ann Arbor, MI 48105}
  \maketitle
} \fi

\if1\blind
{
  \bigskip
  \bigskip
  \bigskip
  \begin{center}
    {\LARGE\bf Title}
\end{center}
  \medskip
} \fi

\bigskip
\begin{abstract}
In this paper, we study a class of problems where the sum of truncated convex functions is minimized. In statistical applications, they are commonly encountered when $\ell_0$-penalized models are fitted and usually lead to NP-Hard non-convex optimization problems. In this paper, we propose a general algorithm for the global minimizer in low-dimensional settings. We also extend the algorithm to high-dimensional settings, where an approximate solution can be found efficiently. We introduce several applications where the sum of truncated convex functions is used, compare our proposed algorithm with other existing algorithms in simulation studies, and show its utility in edge-preserving image restoration on real data.
\end{abstract}

\noindent%
\vfill

\newpage
\spacingset{1.45} 

\section{Introduction}

Regularization methods in statistical modeling have gain popularity in many fields, including variable selection, outlier detection, and signal processing. Recent studies~\citep{shen2012likelihood,she2012outlier} have shown that models with non-convex penalties possess superior performance compared with those with convex penalties. While the latter in general can be obtained with ease by virtue of many well-developed methods for convex optimization~\citep{boyd2004convex}, there are limited options in terms of global solutions for non-convex optimization, which are more and more commonly encountered in modern statistics and engineering. Current approaches often rely on convex relaxation~\citep{candes2010power}, local solutions by iterative algorithms~\citep{fan2001variable} or trading time for global optimality with stochastic search~\citep{zhigljavsky2007stochastic}.

In this paper, we study a special class of non-convex optimization problems, for which the objective function can be written as a sum of truncated convex functions. That is,
\begin{equation}
\label{model}
\bmx=\argmin_\bmx\sum_{i=1}^n\min\{f_i(\bmx),\lambda_i\},
\end{equation}
where $f_i:R^d\rightarrow R, i=1,\ldots,n,$ are convex functions and the truncated levels $\lambda_i\in R, i=1,\ldots,n,$ are constants. Due to the truncation of $f_i(\cdot)$ at $\lambda_i$, the objective function is often non-convex. See Figure~\ref{1d} for an example. 

While in general such problems are NP-Hard (see Section~\ref{sec:algorithm} for formal results), we show that for some $f_i(\cdot)$ there is a polynomial-time algorithm for the global minimizer in low-dimensional settings. The idea is simple: When the objective function is piecewise convex (e.g., see Figure~\ref{1d}), we can partition the domain so that the objective function becomes convex when restricted to each piece. This way, we can find the global minimizer by enumerating all the pieces, minimizing the objective function on each piece, and taking the minimum among all local minima.

\begin{figure}
\begin{center}
\includegraphics[width=0.5\textwidth]{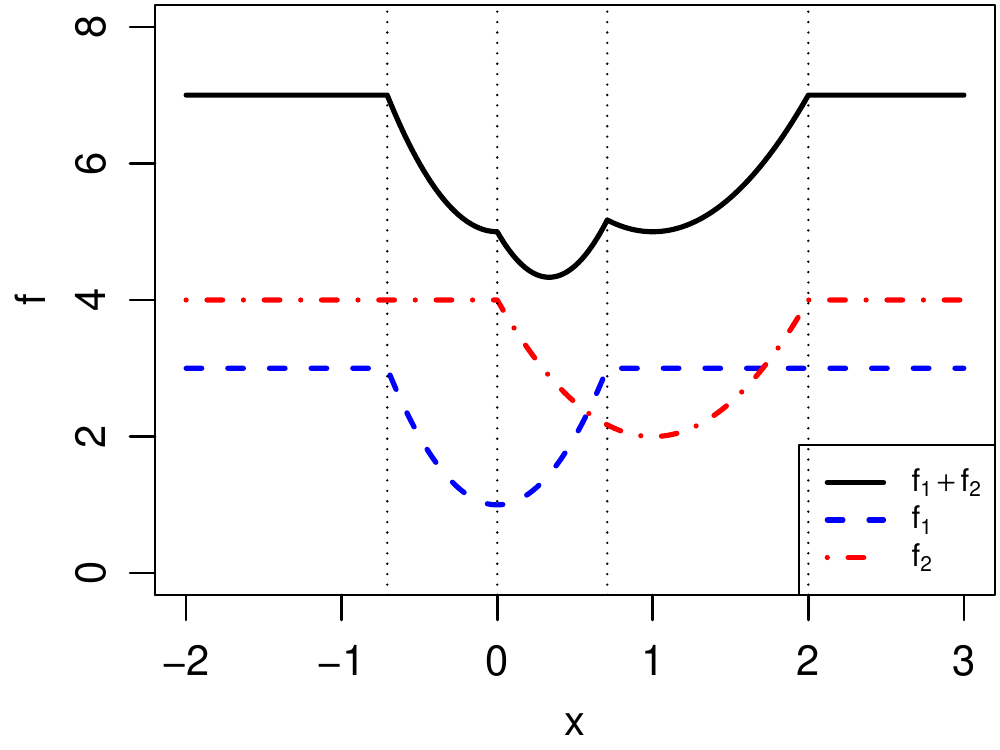}
\end{center}
\caption{The sum of two truncated quadratic functions $f_1+f_2$ (in black), where $f_1(x)=\min\{4x^2+1,3\}$ (in blue) and $f_2(x)=\min\{2(x-1)^2+2,4\}$ (in red). \label{1d}}
\end{figure}

The rest of the paper is organized as follows. In Section~\ref{sec:applications}, we demonstrate the utility of our algorithm in several applications where the objective function can be transformed into a sum of truncated convex functions. In Section~\ref{sec:algorithm}, we lay out the general algorithm for the global solution and its implementation in low-dimensional settings. As we will see in the complexity analysis, the running time grows exponentially with the number of dimensions. We therefore make a compromised but efficient extension of the algorithm in high-dimensional settings. In Section~\ref{sec:experiments}, we compare our proposed algorithm with existing methods in simulation studies, and apply our proposed algorithm to real-life image restoration problems. Discussions are given in Section~\ref{sec:discussion}.

\section{Applications}
\label{sec:applications}

\subsection{Outlier detection in linear models}
\label{sec:outlier.application}

The task of outlier detection in linear regression can be formulated as a problem of variable selection. As in~\citet{gannaz2007robust} and~\citet{mccann2007robust}, given $n$ observations and $p$ covariates, we can add $n$ additional parameters $\{\gamma_i\}_{i=1}^n$ denoting the amount by which the observations are outlying. That is, 
\beq
\label{outlier}
y_i=\bmx_i^T\bmbeta+\gamma_i+\epsilon_i, \quad \quad i=1,\ldots,n,
\eeq
where $y_i\in R, \bmx_i\in R^p, i=1,\ldots,n,$ are the observations, $\bmbeta\in R^p, \gamma_i\in R, i=1,\ldots,n,$ are the parameters of interest, and $\{\epsilon_i\}_{i=1}^n$ are i.i.d. $N(0,\sigma^2)$. Since there are $n+p$ parameters but only $n$ observations, the model is non-identifiable. \citet{gannaz2007robust} used an $\ell_1$ penalty in the objective function to force sparsity in $\gamma$ such that $y_i$ is considered an outlier if $\gamma_i \neq 0$ and an observation conforming to the assumed distribution if $\gamma_i=0$. \citet{mccann2007robust} treated~(\ref{outlier}) as a variable selection problem and applied the Least Angle Regression. Similar idea for outlier detection has also been used for robust Lasso regression~\citep{nasrabadi2011robust,katayama2015sparse}, Poisson regression~\citep{jiang2013penalized}, logistic regression~\citep{tibshirani2014robust}, clustering~\citep{witten2013penalized,georgogiannis2016robust}, as well as a large class of regression and classification problems intoduced in~\citet{lee2012regularization}.

\citet{she2012outlier} took into consideration the issues of masking and swamping when there are multiple outliers in the data. By definition, masking refers to the situation when a true outlier is not detected because of other outliers. Swamping, on the other hand, refers to the situation when an observation conforming to the assumed distribution is considered outlying under the influence of true outliers. They pointed out that using the $\ell_0$ penalty instead of the $\ell_1$ penalty in the objective function could resolve both issues. Assuming $\sigma$ is known, adding an $\ell_0$ penalty to the negative log-likelihood function for model~(\ref{outlier}), the objective function becomes
\beq
\label{likelihood}
f(\bmbeta,\bmgamma)=\sum_{i=1}^n(y_i-\bmx^T_i\bmbeta-\gamma_i)^2+\lambda\sum_{i=1}^n1(\gamma_i\neq0),
\eeq
where $\lambda$ is a tuning parameter and $1(\cdot)$ is the indicator function. It can be shown that this problem can be solved by minimizing a sum of truncated quadratic functions.
\begin{proposition}
\label{equivalence}
Minimizing~(\ref{likelihood}) in $\bmbeta$ and $\bmgamma$ jointly is equivalent to minimizing the following sum of truncated quadratic functions in $\bmbeta$
$$
g(\bmbeta)=\sum_{i=1}^{n} \min\{(y_i-\bmx^T_i\bmbeta )^2, \lambda\}.
$$
\end{proposition}

This result is consistent with the proposition by~\citet{she2012outlier} that the estimate $\hat{\bmbeta}$ from minimizing~(\ref{likelihood}) is an $M$-estimate associated with the skipped-mean loss. Since the objective function is non-convex, \citet{she2012outlier} proposed an iterative hard thresholding algorithm named $\Theta$-IPOD (iterative procedure for outlier detection) to minimize it. Similar to other iterative procedures, $\Theta$-IPOD only guarantees local solutions. A simulation study comparing our proposed algorithm with $\Theta$-IPOD and several other robust linear regression algorithms is presented in Section~\ref{sec:outlier.experiment}. We implement the $\Theta$-IPOD algorithm in R (see Supplementary Algorithm~\ref{alg:IPOD} for details).

Furthermore, Proposition~\ref{equivalence} can be extended to the class of generalized linear models (GLMs). Suppose that $Y_i\in R, i=1,\ldots,n,$ follow a distribution in the exponential family,
$$
f(Y_i=y_i | \theta_i, \phi)=\exp\left\{ \frac{y_i \theta_i -b(\theta_i)}{a(\phi)}+c(y_i, \phi)\right\},
$$
where $\theta_i$ is the canonical parameter and $\phi$ is the dispersion parameter (assumed known here).  For a GLM with canonical link function $g$, $\theta_i=g(\mu_i)=\bmx_i^T\bmbeta+\gamma_i$, the $\ell_0$-penalized negative log-likelihood function is
\beq
\label{glm}
f(\bmbeta,\bmgamma)=\sum^n_{i=1}\{ b(\bmx_i^T\bmbeta+\gamma_i)-(\bmx_i^T\bmbeta+\gamma_i)y_i\}+\lambda\sum_{i=1}^n1(\gamma_i \neq 0).
\eeq
It can be shown that minimizing~(\ref{glm}) is equivalent to minimizing a sum of truncated convex functions.
\begin{proposition}
\label{glm.equivalence}
Minimizing~(\ref{glm}) in $\bmbeta$ and $\bmgamma$ jointly is equivalent to minimizing the following function in $\bmbeta$
$$
g(\bmbeta)=\sum^n_{i=1} \min\{b(\bmx_i^T\bmbeta)-(\bmx_i^T\bmbeta)y_i, \lambda_i^*\},
$$
where $\lambda_i^*=b(g(y_i))-g(y_i)y_i +\lambda, i=1,\ldots,n,$ are constants. Since $b$ is convex~\citep{agarwal2011generative}, the above is a sum of truncated convex function. 
\end{proposition}
\begin{example}
Suppose that $\{Y_i\}_{i=1}^n$ follow Poisson distributions with mean $\{\mu_i\}_{i=1}^n$, respectively, and that $g(\mu_i)=\log\mu_i=\bmx_i^T\bmbeta+\gamma_i$, where $\gamma_i = 0$ if $y_i$ conforms to the assumed distribution and $\gamma_i \neq 0$ if $y_i$ is an outlier. The $\ell_0$-penalized negative log-likelihood function is
\beq
\label{poisson}
f(\bmbeta,\bmgamma) =\sum_{i=1}^n\left\{e^{\bmx_i^T\bmbeta+\gamma_i}-(\bmx_i^T\bmbeta+\gamma_i)y_i\right\}+\lambda\sum_{i=1}^n1(\gamma_i \neq0).
\eeq
According to Proposition~\ref{glm.equivalence}, minimizing~(\ref{poisson}) is equivalent to minimizing the following function
$$
g(\bmbeta)=\sum_{i=1}^n \min \{ e^{\bmx_i^T\bmbeta}-(\bmx_i^T\bmbeta)y_i, \lambda^*_i\}, \text{ where } \lambda^*_i=\lambda-y_i\log y_i+y_i, 
$$
which is a sum of truncated convex functions.
\end{example}
 
\subsection{Convex shape placement} 
\label{sec:shape.application}

Given a convex shape $S\subset R^d$, and $n$ points $\bmp_i\in R^d, i=1,\ldots,n$, each associated with weight $w_i>0$, the problem of finding a translation of $S$ such that the total weight of the points contained in $S$ is maximized has applications in the placement of facilities or resources such as radio stations, power plants or satellites~\citep{mehrez1982maximal}. For some simple shapes (e.g., circles or polygons) in low-dimensional settings, this problem has been well studied~\citep{chazelle1986circle,barequet1997translating}.

We show that this problem can be solved by minimizing a sum of truncated convex functions. Without loss of generality, let $S_0\subset R^d$ denote the region covered by $S$ when it is placed at the origin. Here the location of $S$ can be defined as the location of its centroid. For each point $\bmp_i$, let $S_i\subset R^d$ be the set of locations for placing $S$ such that it covers $p_i$. It is easy to see that $S_i=\{\bmx : \bmp_i-\bmx\in S_0\}=\{\bmp_i-\bmy : \bmy\in S_0\}$, and that the shape of $S_i$ is simply a mirror image of $S_0$ and therefore it is also convex. Furthermore, define convex function $f_i:R^d\rightarrow R$ as 
$$f_i(\bmx)=\left\{\begin{array}{ll}
	-w_i	&\mbox{\quad if }\bmx\in S_i,\\
	\infty	&\mbox{\quad otherwise.}
\end{array}\right.
$$
Then the optimal placement of $S$ can be found by minimizing the sum of truncated convex functions $\sum_{i=1}^n\min\{f_i(\bmx),\lambda_i\}$ as in~(\ref{model}) where $\lambda_i=0, i=1,\ldots,n$.

Some examples of this application are given in Section~\ref{sec:shape.experiment}.

\subsection{Signal and image restoration}
\label{sec:image.application}

Signal restoration aims to recover the original signal from observations corrupted by noise. Suppose that the observed data $\bmy$ are generated from the original data $\bmx$ following the model~\citep{portilla2015efficient}:
$$\bmy=\bmH\bmx+\bmepsilon$$
where $\bmH$ is a matrix performing some linear transformation on the data (e.g., smoothing) and $\bmepsilon$ is the vector the measurement errors, often modeled as additive white Gaussian noise (AWGN). The goal is to estimate (a.k.a. restore or reconstruct) $\bmx$ from observed $\bmy$ and a known $\bmH$. When both $\bmx$ and $\bmy$ are (vectorized) images, the problem is called image restoration.

During this restoration process, one often wants to preserve the edges in the original signal, if there were any. One popular approach is to minimize the following regularized objective function (a.k.a. energy function~\citep{nikolova2011energy}):
$$\hat{\bmx}=\argmin_\bmx{L(\bmH\bmx-\bmy)+\alpha p(\bmx)}$$
where $L(\bmH\bmx-\bmy)$ is the loss function, usually taken as the negative log-likelihood function (e.g., $||\bmH\bmx-\bmy||^2$ in case of Gaussian noise), $p(\bmx)$ is a penalty function to introduce the prior that one wishes to enforce on the original data $\bmx$, and $\alpha$ is a tuning parameter. Many penalty functions have been studied in the literature. While convex penalty functions are generally easier to optimize, non-convex penalty functions can lead to better restoration quality~\citep{nikolova2010fast}. In particular, the truncated quadratic penalty has been found to be quite effective~\citep{nikolova2000thresholding,portilla2015efficient}. For instance, to promote both sharp edges and smooth regions in the estimated $\hat{\bmx}$, a truncated quadratic penalty on the differences between neighboring data points can be used:
$$p(\bmx)=\sum_{i,j\in I, i\in D(j)}\min\{(x_i-x_j)^2, \lambda\},$$
where $I$ is the index set of all the data points (or pixels), and $i\in D(j)$ means that data points (or pixels) $i$ and $j$ are neighbors of each other. Together with this penalty function, the energy function $L(\bmH\bmx-\bmy)+\alpha p(\bmx)$ with the loss function for Gaussian noise is in the form of a sum of truncated quadratic functions, where the loss function $L(\bmH\bmx-\bmy)=||\bmH\bmx-\bmy||^2$ can be regarded as a sum of quadratic functions truncated at infinity. A simulation study comparing our proposed algorithm with other algorithms for signal restoration and an application of our proposed algorithm to image restoration on real data are presented in Section~\ref{sec:image.experiment}.

\section{Methods}
\label{sec:algorithm}

First, the general problem of minimizing a sum of truncated convex functions is in the class of NP-Hard. This can be shown by reducing the 3-satisfiability (3-SAT) problem~\citep{cook1971complexity,karp1972reducibility}, an NP-complete problem, to the problem of minimizing a sum of truncated convex functions. 
\begin{proposition}
\label{3sat}
The 3-SAT problem can be reduced to the problem of minimizing a sum of truncated convex functions. 
\end{proposition}
Consequently, a universal algorithm for solving the general problem of minimizing a sum of truncated convex functions with polynomial running time is unlikely to exist~\citep{michael1979computers}. However, when partitioning the search space such that the objective function is convex when restricted on each region and enumerating all the regions is feasible, a polynomial time algorithms does exist (note that here we consider observations as the input and hold dimensionality of the search space constant). Next, We show that it is in fact the case for some commonly used convex functions in low-dimensional settings.

\subsection{Notations}
\label{sec:notations} 
 
Given $n$ convex functions $f_ i: R^d\rightarrow R, i=1,\ldots,n,$ and constants $\lambda_i\in R$, $i=1,\ldots,n$, we want to find $\bmx\in R^d$ such that the following sum is minimized at $\bmx$
\beq
\label{objective}
f(\bmx)=\sum_{i=1}^{n}\min\{ f_i(\bmx), \lambda_i \} . 
\eeq
Without loss of generality, we further assume $\lambda_i=0$ for all $i$, since minimizing~(\ref{objective}) is equivalent to minimizing 
$$
g(\bmx)=\sum_{i=1}^{n} \min \{g_i(\bmx), 0\} + \sum_{i=1}^{n}\lambda_i. 
$$
where $g_i:R^d\rightarrow R$ is defined as $g_i(\bmx)=f_i(\bmx)-\lambda_i$, which is also convex. Furthermore, we define $C_i\subset R^d$ as the convex region on which $f_i$ is less than or equal to zero,
$$
C_i:=\{\bmx: f_i(\bmx)\leq 0\},
$$
and we define $\partial C_i := \{\bmx: f_i(\bmx)=0\}$, the boundary of $C_i$, as the truncation boundary of $f_i$.
Then, $\{\partial C_i\}_{i=1}^n$, the truncation boundaries of all the $f_i$'s, partition the domain $R^d$ into disjoint pieces $A_1,\ldots,A_m$ such that
$$A_j\cap A_k=\emptyset, \quad\forall j\neq k \quad \mbox{and} \quad \cup_{j=1}^mA_j=R^d,$$
where $A_j$ is defined as
$$A_j=(\underset{k \in I_j}{\cap}C_k) \cap (\underset{l \notin I_j}{\cap}C_l^c), \quad I_j\subset\{1,\ldots,n\},\quad j=1,\ldots,m,$$
where $I_j$ is the index set for a subset of $\{f_1,\ldots,f_n\}$ such that given any $\bmx\in A_j$, $f_k(\bmx)\leq0$ for all $k\in I_j$ and $f_k(\bmx)>0$ for all $k\notin I_j$. An example of partitioning $R^2$ into disjoint pieces $A_1,\ldots,A_m$ is shown in Figure~\ref{fig:example}. The algorithms to find and traverse through all $A_j$'s while constructing the corresponding $I_j$'s will be described in Sections~\ref{sec:framework} and ~\ref{sec:implemention}.

\begin{figure}
\centering
\includegraphics[width=0.5\textwidth]{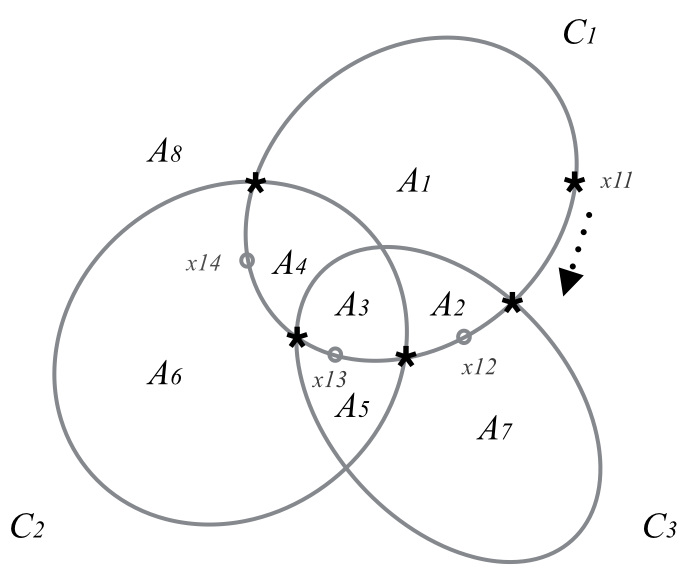}
\caption{The corresponding $C_i$'s of three convex functions $f_1, f_2, f_3$ define on $R^2$, where $C_i=\{\bmx: f_i(\bmx)\leq 0\}$. The boundaries of $\{C_i\}_{i=1}^3$ partition $R^2$ into eight disjoint pieces $\{A_j\}_{j=1}^{8}$.\label{fig:example}}
\end{figure}

\subsection{The general algorithm}
\label{sec:framework}

Our goal is to find the local minimum on each region $A_j$ in the partition and take the minimum of all local minima as the global solution. That is,
$$
\min_\bmx \sum_{i=1}^{n}\min\{f_i(\bmx), 0\} =\min_j\min_{\bmx\in A_j} \sum_{k \in I_j} f_k(\bmx).
$$
To minimize $f(\bmx)$ when restricted to $A_j$, we need to find the index set $I_j$, and minimize $\sum_{k\in I_j}f_k(\bmx)$ subject to $\bmx\in A_j$,
which leads to a series of constrained optimization problems. Although the objective function $\sum_{k\in I_j}f_k(\bmx)$ is a sum of convex functions and therefore is also convex, the domain $A_j$ can be a non-convex set. For instance, except for $A_3$, all other $A_j$'s in Figure~\ref{fig:example} are non-convex sets. Solving such constrained optimization problems can be very challenging. Fortunately, the following proposition shows that it is safe to ignore the constraint $\bmx\in A_j$ when minimizing $\sum_{k\in I_j}f_k(\bmx)$, and consequently, we only need to solve a series of unconstrained convex optimization problems, which is much easier.
\begin{proposition} 
\label{ignore}
Using the notations defined in Section~\ref{sec:notations}, we have
$$
\min_\bmx \sum_{i=1}^{n}\min\{f_i(\bmx), 0\} =\min_j\min_\bmx \sum_{k \in I_j} f_k(\bmx)
$$
\end{proposition}

Based on Proposition~\ref{ignore}, a general framework for minimizing~(\ref{objective}) is to enumerate all the regions $\{A_j\}_{j=1}^m$ and solve a unconstrained convex optimization problem for each region. See Supplementary Algorithm~\ref{alg:general} for details.

\subsection{Implementation in low-dimensional settings}
\label{sec:implemention}

The implementation of the general algorithm described above depends on both the class of functions $\{f_i\}_{i=1}^n$ and the dimension $d$. When $d=1$, each $C_i$ is an interval on the real line and the boundary of $C_i$, $\partial C_i$, is composed of the two end-points of $C_i$, which are the locations where $f_i$ crosses zero. Without loss of generality, assuming that the $2n$ end-points of $\{C_i\}_{i=1}^n$ are all distinct, we can then order them sequentially along the real line which partitions $R$ into $m=2n+1$ fragments $\{A_j\}_{j=1}^m$. We can then go through them one by one sequentially and in the same time keep track of functions entering and leaving the set of untruncated functions on each fragment $A_j$. The detailed procedure for finding the global minimizer of $f(x)$ in 1-D is described in Supplementary Algorithm~\ref{alg:1d}.

When $d=2$, each $C_i$ is a convex region on $R^2$, and its boundary $\partial C_i$ is a curve. One way to enumerate all the $A_j$'s is to travel along each $\partial C_i$, and record the intersection points of $\partial C_i$ and $\partial C_k$ for $ k\neq i$. We then use these intersection points to keep track of functions entering and leaving the set of untruncated functions on each $A_j$. The detailed procedure for finding the global minimizer of $f(\bmx)$ in 2-D is described in Supplementary Algorithm~\ref{alg:2d}.

Using the notations in Section~\ref{sec:notations} and the example in Figure~\ref{fig:example} as an illustration, we start from an arbitrary point $\bmx_{11}$ on $\partial C_1$. On one side we have the region $A_1$, on which there is only one untruncated function ($I_1=\{1\}$). On the other side we have $A_8$, on which every function is truncated ($I_8=\emptyset$). Traveling clockwise, we come across $\partial C_3$. At this point, we add $f_3$, which gives the sets of untruncated functions on $A_2$ ($I_2=\{1,3\}$) and $A_7$ ($I_7=\{3\}$). Similarly, we obtain $I_3=\{1,2,3\}$ and $I_5=\{2,3\}$ when we come across $\partial C_2$. When we come acoss $\partial C_3$ for the second time, we remove $f_3$ from the set of untruncated function and obtain $I_4=\{1,2\}$ and $I_6=\{2\}$. By repeating the process for all $C_i$'s, we enumerate the set of untruncated functions on all $A_j$'s.

What remains to be supplied in the 1-D algorithm are methods to find the end-points of any given $C_i$, and to minimize the sum of a subset of untruncated functions. Similarly, for the 2-D algorithm we need ways to find the intersection points of any given $\partial C_i$ and $\partial C_k$, and to minimize the sum of a subset of untruncated functions. The implementation of these steps depends on the class of functions that we are dealing with. For some function classes, solutions for these steps are either straightforward, or already well-studied. For instance, for quadratic functions, finding the end-points (in 1-D) or finding the intersections (in 2-D) requires solving quadratic equations, for which closed-form solutions exist. Minimizing the sum of a subset of quadratic functions can also be solved in closed-form. For convex shape placement problem described in Section~\ref{sec:shape.application}, published algorithms exist for these steps for commonly encountered convex shapes such as circles or convex polygons~\citep{de2000computational}. For more general convex functions (e.g., those described in Section~\ref{sec:outlier.application} for GLMs), iterative algorithms (e.g., gradient descent or the Newton-Raphson method) can be used for these steps.

\subsection{Extension to high-dimensional settings}
\label{sec:high.dimension}

In three or higher dimensions, our algorithm can be implemented by following the same idea of tracking all the intersection points as in the 2-D case. Essentially, each boundary $\partial C_i$ is a $d-1$ dimensional surface, and enumerating all the $A_j$'s can be achieved by traversing through all the pieces on each $\partial C_i$ that are formed by its intersections with all other $\partial C_k$'s, which is in turn a $d-1$ dimensional problem. For instance, when $d=3$, we need to find all the intersection curves of $\partial C_i$ and $\partial C_k$ (both of which are surfaces) for $i\neq k$, and traverse along each intersection curve while keep tracking all other surfaces $\partial C_j, j\neq i\neq k,$ it crosses. Apparently, this algorithm becomes increasingly complicated and inefficient for larger $d$, which renders it impractical. 

Here, we propose a compromised but efficient extension of our proposed algorithm to high-dimensional settings. The price we pay is to give up the global minimizer, which is sensible choice as Proposition~(\ref{3sat}) has shown that the general problem is NP-Hard. In particular, we propose to solve for an approximate solution using a cyclic coordinate descent algorithm, where we optimize one parameter a time while keeping all other parameters fixed, and cycle through all the parameters until converge. When restricting to only one parameter, the objective function is simply a sum of truncated convex functions in 1D. Therefore, we can use our 1-D algorithm to solve this subproblem in each iteration. This algorithm is guaranteed to converge since the objective function is bounded below and its value is descending after each iteration. See Supplementary Algorithm~\ref{alg:iterative} for details. We will evaluate the performance of this algorithm using both simulated and real data experiments in Section~\ref{sec:image.experiment}.

\subsection{Time complexity analysis}
\label{sec:complexity}

For time complexity analysis of our proposed algorithms, in low-dimensional settings, we can regard the dimension $d$ as a constant. That is, any univariate function of $d$ can be considered as $O(1)$.

For the 1-D algorithm, finding the $2n$ end-points takes $O(nS)$ time, where $S$ is the time for finding the two endpoints of a given function. Ordering the $2n$ end-points takes $O(n\log n)$ time. Traversing through all the end-points takes $O(nT)$ time, where $T$ is the time for minimizing the sum of a subset of untruncated functions. Similarly, for the 2-D algorithm, finding all the intersection points takes $O(n^2S)$ time, where $S$ is the time for finding all the intersection points of any two given functions. Sorting all the intersection points along all the boundaries $\{\partial C_i\}_{i=1}^n$ takes $O(n^2K\log(nK))$ time, where $K$ is the maximum number of intersection points any two boundaries $\partial C_i$ and $\partial C_j$ can have. Traversing through all the intersection points takes $O(n^2KT)$ time.

First, we show that $K=O(1)$ for a large class of truncated convex functions. That is, given any two truncated convex functions in the class, the maximum number of intersection points their boundaries can have is bounded by a constant.
\begin{definition}
For any positive integer $k\in Z^+$, a class of curves $\mathcal{C}$ in $R^2$ is said to be $k$-intersecting if and only if for any two distinct curves in $\mathcal{C}$, the number of their intersection points is at most $k$. 
\end{definition}
\begin{definition}
A class of truncated functions in $R^2$ is said to be $k$-intersecting if and only if the set of their truncation boundaries is $k$-intersecting. 
\end{definition}
\begin{example}
\label{ex:quad1}
The class of truncated quadratic functions in $R^2$ with positive definite Hessian matrices is $k$-intersecting with $k=4$. This is easy to see given the facts that the truncation boundary of a quadratic function in $R^2$ with positive definite Hessian matrix is an ellipse, and two distinct ellipses can have at most four intersection points.
\end{example}
In fact, according to B\'ezout's theorem, the number of intersection points of two distinct plane algebraic curves is at most equal to the product of the degrees of the corresponding polynomials. Therefore, a class $\mathcal{F}$ of truncated bivariate polynomials is $k^2$-intersecting if for any function $f\in\mathcal{F}$ its untruncated version is a polynomial of degree at most $k$.

While $S$ and $T$ depend on the class of functions that we are dealing with, for some function classes, we have $S=O(1)$ and $T=O(1)$. That is, they both take constant time.
\begin{example}
\label{ex:quad2}
For quadratic functions with positive definite Hessian matrices, $T=O(1)$. This is easy to see given the following three facts:
\begin{enumerate}
\item Given $n$ quadratic functions $f_i=\frac12\bmx^T\bmA_i\bmx+\bmb_i^T\bmx+c_i, i=1,\ldots,n$, their sum is $\sum_if_i(\bmx)=\frac12\bmx^T\bmA\bmx+\bmb^T\bmx+c$, where $\bmA=\sum_i\bmA_i, \bmb=\sum_i\bmb_i$, and $c=\sum_ic_i$, which is also a quadratic function.
\item To update the sum of quadratic functions when adding a new function to the sum or removing an existing function from the sum, we only need to update $\bmA, \bmb$ and $c$, which takes $O(1)$ time (it is in fact $O(d^2)$ time but can be simplified as $O(1)$ time since we consider $d$ as a constant in low-dimensional settings).
\item The minimizer of any quadratic function $\frac12\bmx^T\bmA\bmx+\bmb^T\bmx+c$ with positive definite Hessian matrix is $-\bmA^{-1}\bmb$, which takes $O(1)$ time to compute (it is in fact $O(d^3)$ time but can be simplified as $O(1)$ time since we consider $d$ as a constant in low-dimensional settings).
\end{enumerate}
Furthermore, $S=O(1)$, since all the intersection points (up to four of them) of any two given ellipses can be found using closed-form formulas~\citep{richter2011perspectives}.
\end{example}
Putting Examples~\ref{ex:quad1} and~\ref{ex:quad2} together, we know that the running time of the 1-D algorithm for sum of truncated quadratic functions with positive definite Hessian matrix is $O(n\log n)$, and the running time of the 2-D Algorithm for sum of truncated quadratic functions with positive definite Hessian matrix is $O(n^2\log n)$. The time complexity analysis for other class of functions can be conducted similarly.

In high-dimensional settings, however, the running time of the general algorithm will be at least $O(n^d\log n)$, where $d$ is the dimension. In another word, the running time grows exponentially as the dimension increases, which is typical for NP-Hard problems. It is easy to see that the running time of the cyclic coordinate descent algorithm is $O(kdn\log n)$, where $k$ is the number iterations to converge, and $O(dn\log n)$ is the time for each round of $d$ one-dimensional updates.

\section{Experiments}
\label{sec:experiments}

\subsection{Outlier detection in simple linear regression} 
\label{sec:outlier.experiment}

We simulate data for outlier detection in simple linear regression as described in Section~\ref{sec:outlier.application} and compare the performance of our proposed method with the $\Theta$-IPOD algorithm~\citep{she2012outlier} and three other robust estimation methods: MM-estimator~\citep{yohai1987high}, least trimmed squares (LTS)~\citep{rosseeuw1987robust} and~\citet{gervini2002class} one-step procedure (denoted as GY). Our goal is to estimate the regression coefficients and identify the outliers with $\sigma$ assumed to be 1. In other words, we try to estimate $\bmbeta$ and $\bmgamma$ in (\ref{outlier}). Given $n$ observations and $k$ outliers, let $\bmX=[\bmone_n, (\bmx_1,\ldots,\bmx_n)^T]$,  $\bmbeta=(\beta_0, \beta_1)^T=(1,2)^T$, and $L$ be a parameter controlling the leverage of the outliers. When $L>0$, $x_i$ is drawn from $uniform(L,L+1)$ for $i=1,\ldots,k$, and from $uniform(-15,15)$ for $i=k+1,\ldots,n$. $\bmgamma=(\gamma_1,\ldots,\gamma_n)^T$ represents deviations from the means, and each $\gamma_i$ is drawn from $exponential(0.1)+3$ for $i=1,\ldots,k$, and $\gamma_i=0$ for $i=k+1,\ldots,n$. Based on a popular choice for $\sqrt{\lambda}$ as $2.5\hat{\sigma}$~\citep{she2012outlier,wilcox2005robust,maronna2006robust}, we set $\sqrt{\lambda}$ as $2.5$. 

We simulate $100$ independent data sets, each with $100$ observations (i.e., $n=100$). The results are shown in Figure~\ref{fig:outlier} and Supplementary Table~\ref{tab:outlier}. The performance of each method is evaluated by the masking probability and the swamping probability under two scenarios: (i) No $L$ applied (denotes as $L=0$), that is, $x_i$ is drawn from $uniform(-15,15)$ for $i=1,\ldots,n$, and (ii) $L=20$. Masking probability, as in~\citet{she2012outlier}, is defined as the proportion of undetected true outliers among all outliers. Swamping probability, on the other hand, is the fraction of normal observations recognized as outliers. We can see that the proposed method outperforms others, especially when the number of outliers is high.

\begin{figure}
\begin{center}
\includegraphics[width=0.8\textwidth]{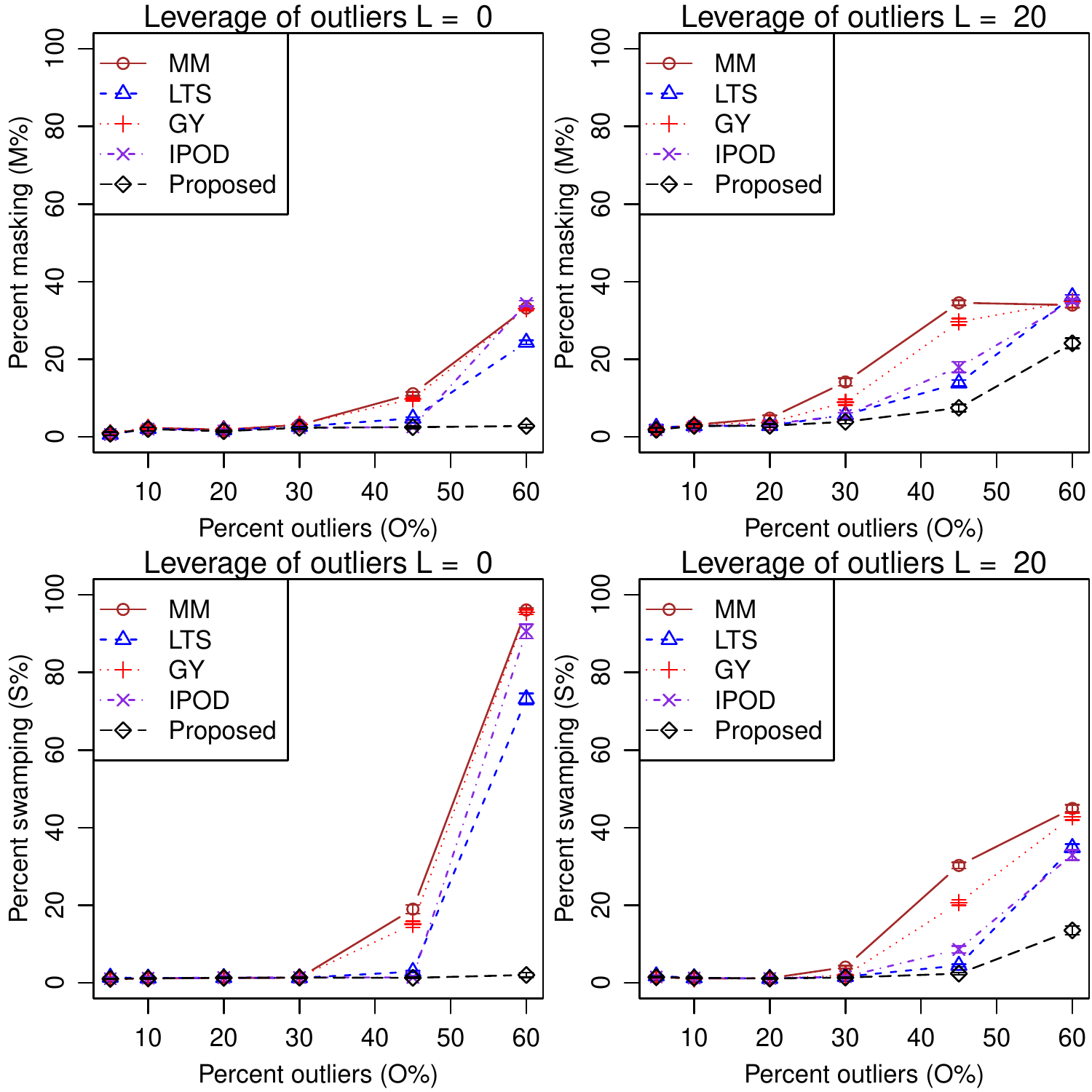}
\end{center}
\caption{Comparison of different methods for outlier detection in simple linear regression. The figures show the mean percents of masking (top) and swamping (bottom) for different leverages of outliers: $L=0$ (left) and $L=20$ (right) and differnt percents of outliers (O\%) for all the methods using 100 simulated replicates. The standard errors of the means are shown as error bars. \label{fig:outlier}}
\end{figure}

\subsection{Sum of truncated quadratic functions}
\label{sec:quadratic.experiment}
 
We simulate sum of truncated quadratic functions with positive definite Hessian matrix in $R^2$ and compare the performance of the proposed algorithm with several other competing algorithms including a global search algorithm (the DIRECT algorithm)~\citep{jones1993lipschitzian} and a branch-and-bound global optimization algorithm (StoGO)~\citep{madsen1998new} both implemented in R package \textsf{nloptr}, a generalized simulating annealing algorithm (SA) implemented in R package \textsf{GenSA}~\citep{xiang2013generalized}, a particle swarm optimization algorithm (PSO) implemented in R package \textsf{hydroPSO}~\citep{zambrano2013model}, as well as the difference of convex functions (DC) algorithm~\citep{an1997solving} which has been used to solve problems with truncated convex functions~\citep{shen2012likelihood,chen2016personalized}. We implement the DC algorithm in R (see Supplementary Section~\ref{sec:DC} for details). 

Following~\citep{hendrix2010introduction}, we compare the performance of all the algorithms in terms of their effectiveness in finding the global minimum. We measure effectiveness by the success rate, where a success for a given algorithm in a given run is defined as having the estimated minimum no greater than any other algorithms by $10^{-5}$. This tolerance value is allowed to accommodate numerical precision issues. We set a maximum number of $10^4$ function evaluations, a maximum number of $10^4$ iterations and a convergence tolerance level of $10^{-8}$ for all competing algorithms whenever possible. See Supplementary Table~\ref{tab:settings} for details.

We randomly generate truncated quadratic functions in $R^2$ with varying degrees of complexity. Specifically, given a quadratic function with positive definite Hessian matrix in $R^2$ truncated at zero, the truncation boundary is an ellipse. Let $a$ and $b$ be the lengths of the two axes of the ellipse, $u$ and $v$ be the x and y coordinates of the center of the ellipse, $\theta$ be the angle between the long axis of the ellipse and the x axis, and $-z$ be the lowest value of the function. For simplicity, we use a single tuning parameter $C$ to control the complexity of the objective function. The larger the $C$, the more local minima the objective function will have. Examples of objective functions with different values of $C$ are given in Figure~\ref{fig:C}. In particular, we randomly sample $\theta$ from $uniform(0, \pi)$, $a$ from $uniform(0.01, 0.5)/C$, $b$ from $uniform(0.01, 0.5)$, $u$ and $v$ from $uniform(0,1)$ and $z$ from $uniform(-10,-1)$. We simulate three scenarios where $C$ is $1$, $5$, and $10$, respectively, and we compute the coefficients of the corresponding quadratic functions based on the above six parameters. For each value of $C$, we simulate $100$ independent data sets each with $50$ random quadratic functions (i.e., $n=50$) truncated at $\lambda=0$.

\begin{figure}
\begin{center}
\includegraphics[width=1\textwidth]{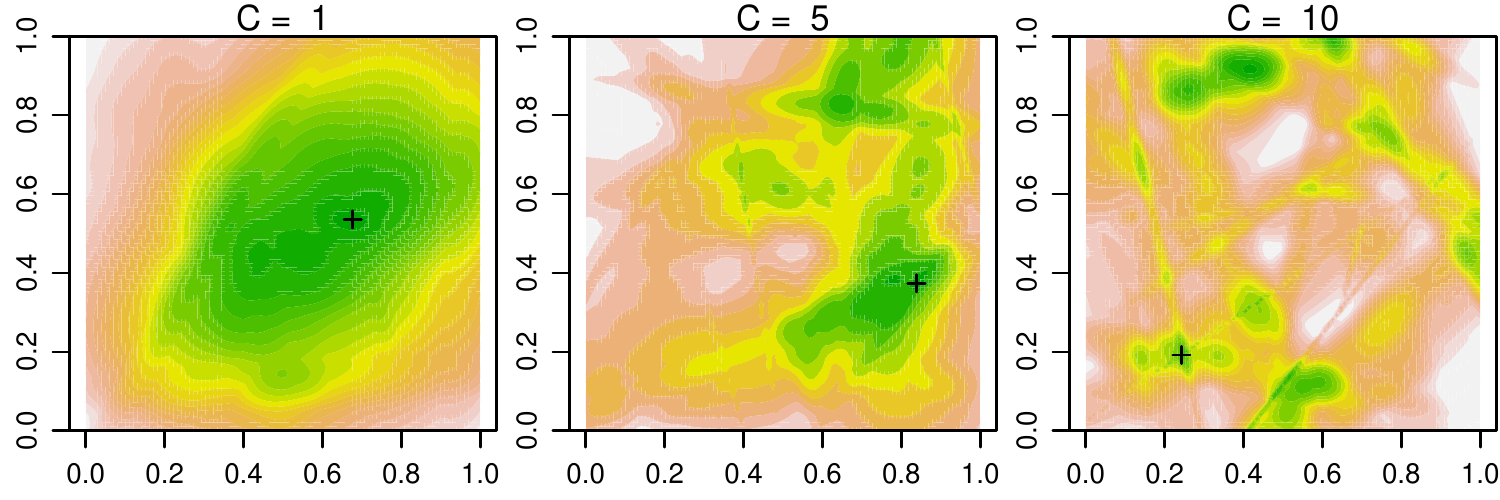}
\end{center}
\caption{Contour plots of randomly generated sum of truncated quadratic functions in $R^2$. Global minima are marked with the plus sign. \label{fig:C}}
\end{figure}

The performance of the proposed algorithm and other competing algorithms are shown in Figure~\ref{fig:quadratic2D} and Supplementary Table~\ref{tab:quadratic.2D}. We can see that our proposed algorithm has a success rate of $100\%$ regardless the value of $C$, as it guarantees to find the global minimizer. For all other competing algorithms, their success rates decline when $C$ increases.

\begin{figure}[h]
\begin{center}
\includegraphics[width=0.65\textwidth]{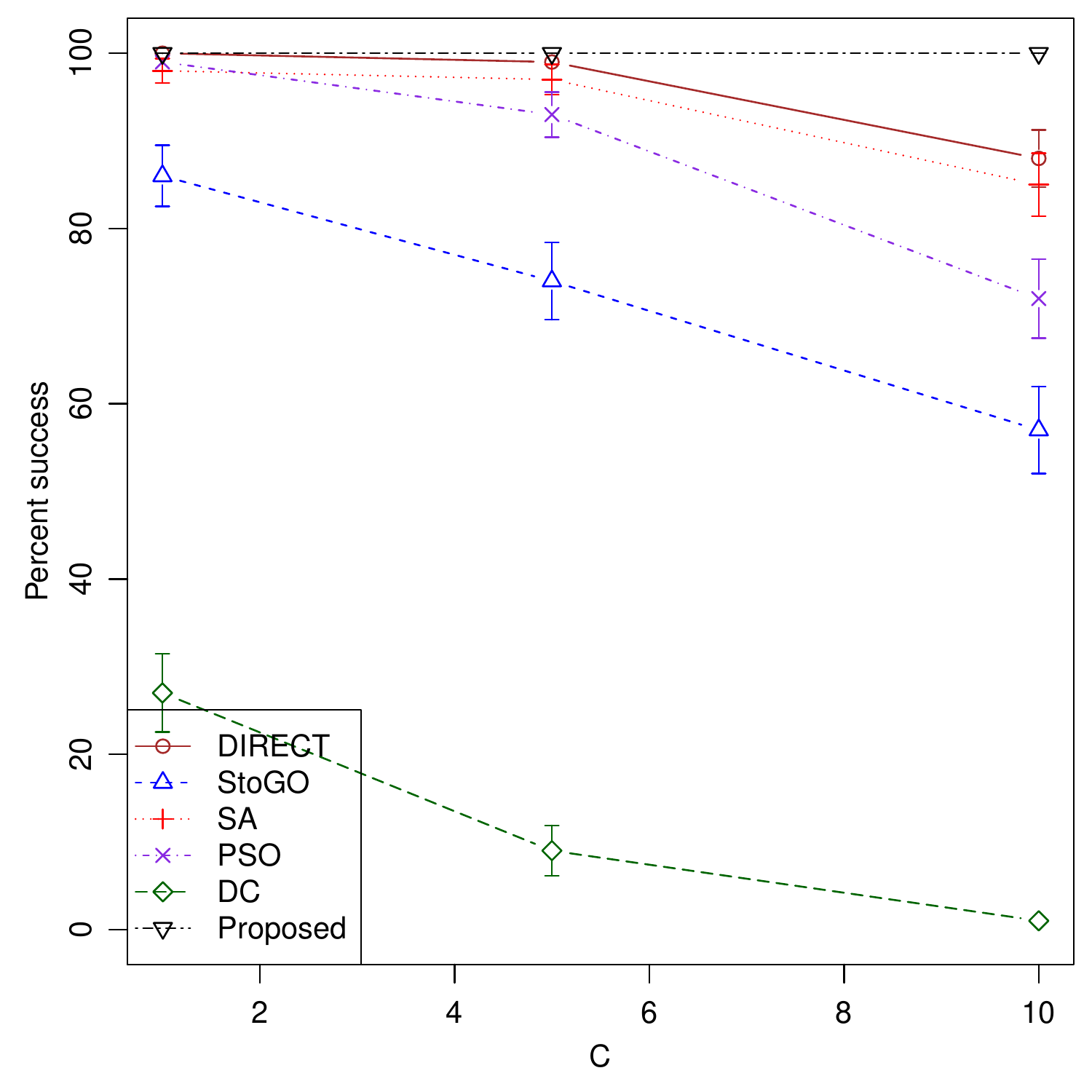}
\end{center}
\caption{Comparison of different algorithms for minimizing the sum of $50$ randomly generated truncated quardratic funstions in 2-D. The figure shows the mean success rates (in percents) for all the methods using 100 simulated replicates for different complexities of the functions ($C$). The standard errors of the means are shown as error bars. \label{fig:quadratic2D}}
\end{figure}

\subsection{Convex shape placement}
\label{sec:shape.experiment}

Following Section~\ref{sec:shape.application}, we randomly sample $30$ points (i.e., $n=30$) uniformly from the $[0,1]\times[0,1]$ unit square, and use our proposed algorithm to find a location to place $S$ such that it covers the maximum number of points. To demonstrate the generality of our proposed algorithm, we consider three shapes here: circle, square and hexagon. The results are shown in Supplementary Figure~\ref{fig:shape}.

\subsection{Signal and image restoration}
\label{sec:image.experiment}

Following Section~\ref{sec:image.application}, we simulate 1-D signal with additive Gaussian noise, and compare the performance of the proposed algorithm with several other algorithms including DIRECT, StoGO, SA, PSO (See Section~\ref{sec:quadratic.experiment} for more details of these algorithms) and a recently published iterative marginal optimization (IMO) algorithm~\citep{portilla2015efficient}, which was specifically designed for signal and image restoration. We implement the IMO algorithm in R (see Supplementary Section~\ref{sec:IMO} for details). The DC algorithm turns out to be numerically equivalent to the IMO algorithm, but much slower. Therefore, we did not included the DC algorithm in the comparison, and simply named the IMO algorithm as IMO/DC. 

The data are simulated by adding random Gaussian noise sampled i.i.d. from $N(0,1)$ to an underlying true signal. Each data set contains $100$ data points equally spaced on the interval $[0,1]$. The true signal is design to be piece-wise smooth with different pieces being constant, linear, quadratic or sine waves (see Figure~\ref{fig:signal}). All the algorithms are used to restore the signal by minimizing the following objective function,
$$\hat{\bmy}=\argmin_{\hat{\bmy}}\sum_{i=1}^d(\hat{y}_i-y_i)^2+w\sum_{i=1}^{d-1}\min\{(\hat{y}_i-\hat{y}_{i+1})^2, \lambda\},$$
where $d=100, y_i$ and $\hat{y}_i, i=1,\ldots,d,$ are the observed and restored values at data point $i$, respectively. That is, we are solving the sum of $199$ truncated quadratic functions ($99$ of them are truncated at $\lambda$, and the remaining $100$ of them are truncated at infinity) in a $100$-dimensional parameter space. The tuning parameters are empirically set as $w=4$ and $\lambda = 9$, respectively. 

\begin{figure}
\begin{center}
\includegraphics[width=0.8\textwidth]{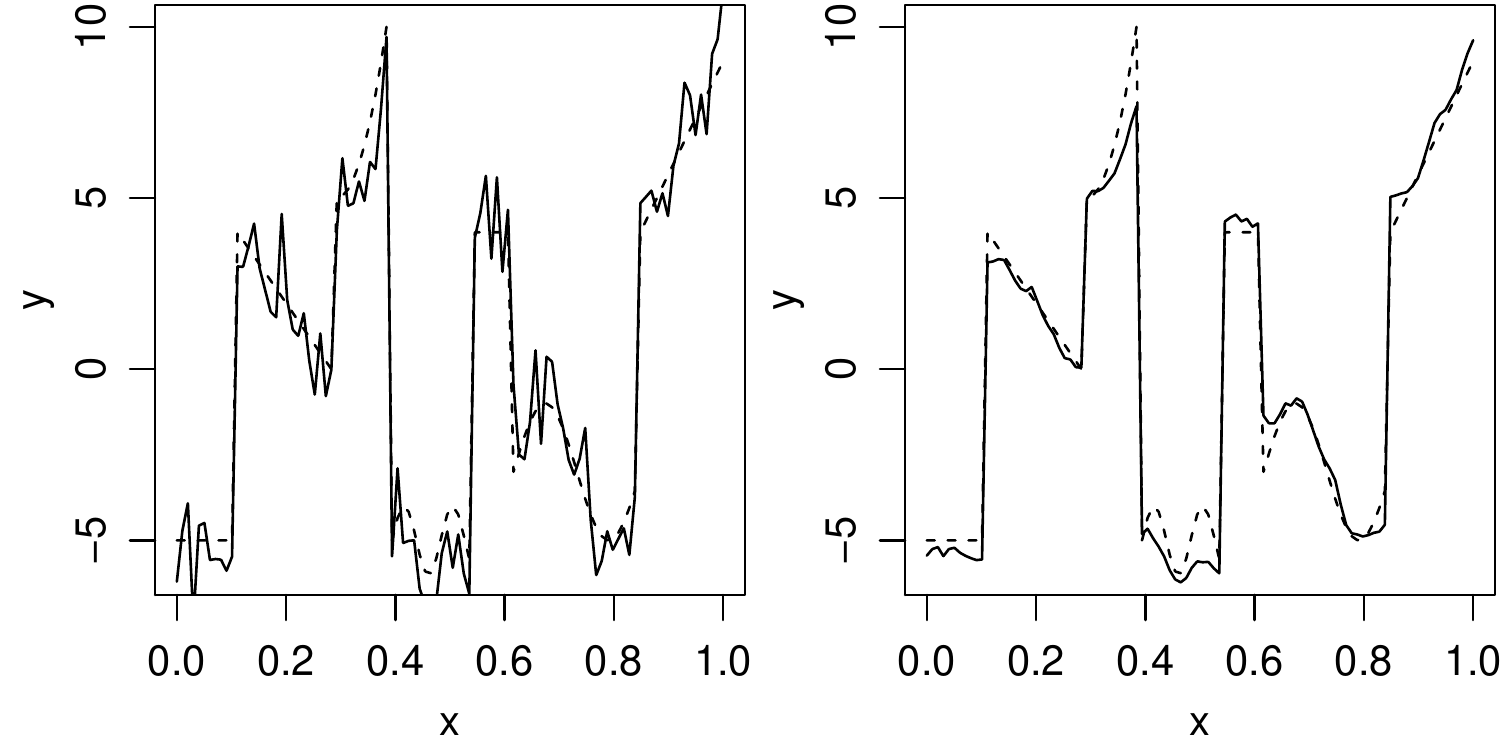}
\end{center}
\caption{Simulated random signal (left) and restored signal (right) are shown in solid lines. The underlying true signal are shown in dashed lines. \label{fig:signal}}
\end{figure}

We measure the performance of these algorithms using four different metrics:
\begin{enumerate}
	\item Success rate, which is defined in Section~\ref{sec:quadratic.experiment}. Note a success here only means that a given algorithm has found the best solution among all algorithms, which may or may not be the global minimizer.
	\item Relative loss, which is defined as $|f(\hat{\bmy})-f(\bmy^*)|/|f(\bmy^*)|$, where $\hat{\bmy}$ and $\bmy^*$ are the solution found by a given algorithm and the best solution found by all algorithms, respectively.
	\item Root mean square error (RMSE), which is defined as $\sqrt{d^{-1}\sum_{i=1}^d(\hat{\bmy}_i-\tilde\bmy_i)^2}$, where $\hat{\bmy}$ and $\tilde\bmy$ are the solution found by a given algorithm and the underlying true signal, respectively.
	\item Running time, measured in seconds.
\end{enumerate}
The performance of the proposed algorithm and other competing algorithms are summarized in Table~\ref{tab:signal}. In general, the proposed algorithm outperforms all other methods in terms of success rate, relative loss and RMSE. It is also significantly faster than all other algorithms.

\begin{table}
\caption{Comparison of different algorithms for signal restoration. The table shows the mean success rates (in percents), relative losses, root mean square errors (RMSE), as well as running times (in seconds) for all the methods using 100 simulated replicates. The standard errors of the means are given in parentheses. \label{tab:signal}}
\centering
\begin{tabular}{ccccccc}
  \hline
 & DIRECT & StoGO & SA & PSO & IMO/DC & Proposed \\ 
  \hline
Success rate &  0.0 (0.0) &  8.0 (2.7) & 52.0 (5.0) &  0.0 (0.0) & 28.0 (4.5) & 84.0 (3.7) \\ 
  Relative loss & 0.08 (0.00) & 0.05 (0.02) & 0.04 (0.01) & 0.17 (0.01) & 0.10 (0.01) & 0.01 (0.00) \\ 
  RMSE & 0.66 (0.01) & 0.56 (0.01) & 0.59 (0.01) & 0.63 (0.01) & 0.56 (0.01) & 0.55 (0.01) \\ 
  Time &  0.40 (0.00) & 61.29 (0.27) &  0.31 (0.00) & 12.39 (0.13) &  1.50 (0.07) &  0.04 (0.00) \\ 
   \hline
\end{tabular}
\end{table}

Finally, we apply the proposed algorithm for image restoration. Both synthetic and real images are used for this experiment (see Figure~\ref{fig:image} and Supplementary Figure~\ref{fig:image.more}). All images are resized to $256\times256$, converted to gray scale and normalized to have pixel intensity levels in $[0,1]$. Independent Gaussian noise sampled from $N(\mu=0, \sigma^2=0.01)$ is added to each pixel, and the proposed algorithm is used to restore the original image via minimizing the following objective function,
$$\hat{\bmz}=\argmin_{\hat{\bmz}}\sum_{i\in I}(\hat{z}_i-z_i)^2+w\sum_{i,j\in I, i\in D(j)}\min\{(\hat{z}_i-\hat{z}_j)^2, \lambda\},$$
where $z_i$ and $\hat{z}_i, i\in I,$ are the observed and restored intensity values at pixel $i$, respectively, $i\in D(j)$ means that pixels $i$ and $j$ are neighbors of each other, and the tuning parameters are empirically set as $w=2$ and $\lambda = 0.02$, respectively. From Figure~\ref{fig:image} and Supplementary Figure~\ref{fig:image.more}, we see that compared with Gaussian smoothing, the proposed algorithm can restore the smoothness in the image while maintaining the sharp edges. Even though this problem has a dimension of $d=256\times256=65,536$ and the number of truncated quadratic functions is $n=256\times256+2\times256\times255=196,096$, it only takes about $10$ seconds for our algorithm to converge.

\begin{figure}
\begin{center}
\includegraphics[width=1\textwidth]{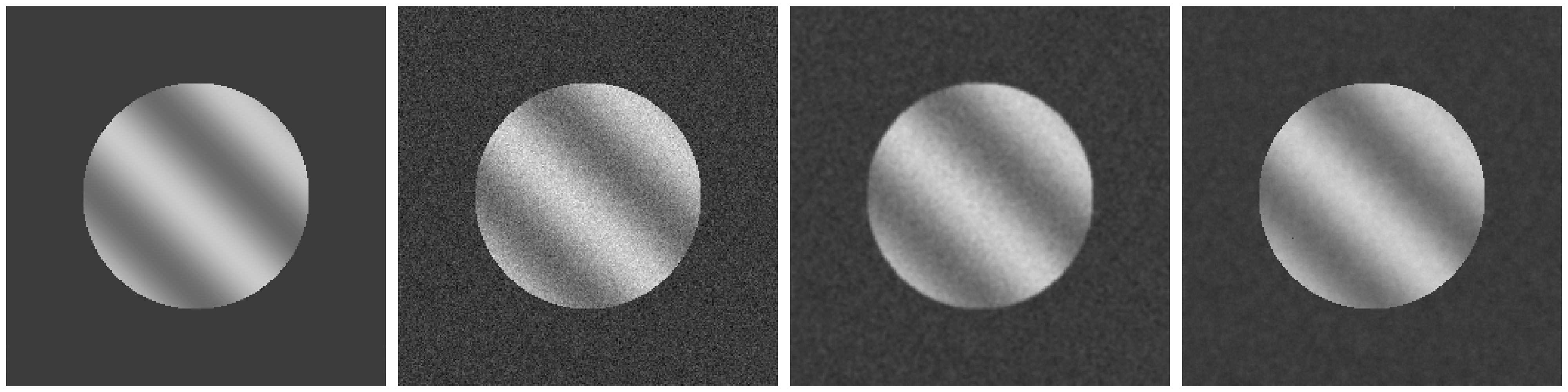}
\includegraphics[width=1\textwidth]{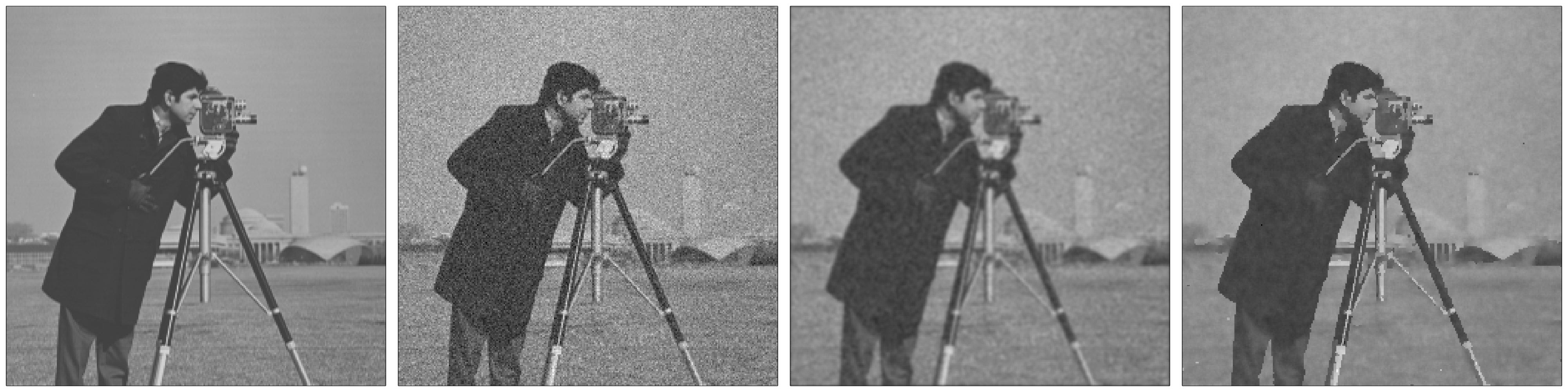}
\end{center}
\caption{Restoration of synthetic and real images. For each row, from left to right: original image, image with Gaussian noise added, image restored using Gaussian smoothing with a $5\times5$ kernel and image restored using proposed algorithm. \label{fig:image}}
\end{figure}

\section{Discussion}
\label{sec:discussion}

We know that summing convex functions together still gives us a convex function. Although simply truncating the function at a given level does not seem to add much complexity to a convex function, the sum of truncated convex functions is not in the same class as its summands, which makes it very powerful and flexible in modeling various kinds of problems, as several examples given in Section~\ref{sec:applications}. Figure~\ref{fig:C} further demonstrates the diverse landscape that can be achieved by a sum of truncated quadratic functions. This flexibility is supported by Proposition~\ref{3sat}, which implies that any problem in the class of NP can be reduced to the minimization of a sum of truncated convex functions. A potential future work is to approximate a given non-convex function by a sum of truncated quadratic functions and then use our proposed algorithm to minimize it.

In the cyclic coordinate descent algorithm, instead of performing a univariate update in each round, we can also perform a bivariate update in each round using the 2-D algorithm (i.e., using a block coordinate descent algorithm), which may help increase the chance of finding the global minimizer, at the cost of more intensive computation.

Besides the applications described in this paper, minimizing sum of truncated convex functions also has many other applications, such as detecting differential gene expression~\citep{jiang2014unit} (See Supplementary Section~\ref{sec:DE.application}) and personalized dose finding~\citep{chen2016personalized}. This paper demonstrates that the proposed algorithm can be quite efficient when the truncation boundaries of the class of convex functions are simple shapes such as ellipse and convex polygon, which cover the cases of truncated quadratic functions and truncated $\ell_1$ penalty (TLP~\citep{shen2012likelihood}). Although these functions are seemingly limited, their applications are vastly abundant, and we have shown only a few selected examples in this paper. In our future work, we will investigate the application of our proposed algorithm to other classes of convex functions.

R programs for reproducing the results in this paper are available at \url{http://www-personal.umich.edu/~jianghui/stcf/}.

\section*{Supplementary materials}

Supplementary texts, algorithms, proofs, figures and tables. (supplementary.pdf)

\section*{Acknowledgements}

We thank the two anonymous reviewers and the associate editor for their suggestions on the image restoration application and the extension to high-dimension settings. Their comments and suggestions have helped us improve the quality of this paper substantially.

\bibliographystyle{Chicago}
\bibliography{manuscript}

\clearpage

\begin{center}
\textbf{\large Supplementary Materials for ``Minimizing Sum of Truncated Convex Functions and Its Applications''}
\end{center}

\newcommand{\beginsupplement}{%
        \setcounter{section}{0}
        \renewcommand{\thesection}{S\arabic{section}}%
        \setcounter{table}{0}
        \renewcommand{\thetable}{S\arabic{table}}%
        \setcounter{figure}{0}
        \renewcommand{\thefigure}{S\arabic{figure}}%
        \setcounter{algorithm}{0}
        \renewcommand{\thealgorithm}{S\arabic{algorithm}}%
        \setcounter{equation}{0}
        \renewcommand{\theequation}{S\arabic{equation}}%
     }

\beginsupplement

\section{Supplementary texts}

\subsection{Application on detecting differential gene expression with $\ell_0$-penalized models}
\label{sec:DE.application}

The idea of using the $\ell_0$ penalty for variable selection can also be applied to the detection of differentially expressed genes from RNA sequencing data. The problem is discussed in detail in~\citet{jiang2014unit}, and we briefly summarize the approach here. Given $S$ experimental groups each with $n_s$ biological samples, we would like to compare the expression levels of $m$ genes measured in the samples. Let $\mu_{si}$ be the mean expression level of gene $i$ (on the log-scale) in group $s$, $d_{sj}$ be the scaling factor (e.g., sequencing depth or library size on the log-scale) for sample $j$ in group $s$, and $\sigma^2_i$ be the variance of expression level of gene $i$ (on the log-scale). Assuming a linear model on the observed data $x_{sij} \sim N(\mu_{si} + d_{sj}, \sigma^2_{i})$, the problem is to identify genes that are differentially expressed across the groups. To do so, assuming $\{\sigma_i\}_{i=1}^m$ are known, reparametrizing $\mu_{si}$ as $\mu_i=\mu_{1i}, \gamma_{si}=\mu_{si}-\mu_{1i}, s=1,\ldots,S$, the $\ell_0$-penalized  negative log-likelihood function of the model is
\beq
\label{DE}
f(\mu,\gamma,d)=\displaystyle\sum_{i=1}^m\frac1{2\sigma_i^2}\sum_{s=1}^S\sum_{j=1}^{n_s}(x_{sij}-\mu_i-\gamma_{si}-d_{sj})^2+\sum_{i=1}^m\alpha_i1(\sum_{s=1}^S|\gamma_{si}|>0)
\eeq
Where $\{\alpha_i\}_{i=1}^m$ are tuning parameters. It is shown in~\citet{jiang2014unit} that~(\ref{DE}) can be solved as follows
$$
\begin{array}{l}
	d^\prime_{sj} = (\sum_{i=1}^m(x_{sij}-x_{si1})/\sigma_i^2)/(\sum_{i=1}^{m}1/\sigma_i^2), s=1,\ldots,S\\
	\mu^\prime_{si} = (1/n_s)\sum_{j=1}^{n_s}(x_{sij}-d^\prime_{sj}) , s=1,\ldots,S\\
	d_1=0\\
        d_2,\ldots,d_S=\displaystyle\argmin_{d_2,\ldots,d_S}\sum_{i=1}^m\min\left(g(d_2,\ldots,d_S), \alpha_i\right)\\
        \mbox{where } \displaystyle g(d_2,\ldots,d_S)=\frac1{2\sigma_i^2}\left\{\sum_{s=1}^Sn_s(\mu_{si}^\prime-d_s)^2-\frac1n\left[\sum_{s=1}^S(n_s(\mu_{si}^\prime-d_s))\right]^2\right\}\\
	d_{sj}=d_s+d^\prime_{sj} , s=1,\ldots,S\\
	$$
	\gamma_{si} = \left\{
	\begin{array}{ll}
		0  &\mbox{ if } g(d_2,\ldots,d_S)<\alpha_i\\
		\mu^\prime_{si}-\mu^\prime_{1i}-d_s &\mbox{ otherswise}		
	\end{array}\right.
	$$ \\
	\mu_i = \left\{
	\begin{array}{ll}
		$$(1/n)\sum_{s=1}^Sn_s(\mu^\prime_{si}-d_s)$$ &\mbox{ if }  g(d_2,\ldots,d_S)<\alpha_i\\
		\mu^\prime_{1i} &\mbox{ otherwise}\\
	\end{array}\right.\\
\end{array}
$$
where the only computationally intensive step is to minimize a sum of truncated quadratic functions in $d_2,\ldots,d_S$
$$ d_2,\ldots,d_S=\displaystyle\argmin_{d_2,\ldots,d_S}\sum_{i=1}^m\min\{g(d_2,\ldots,d_S), \alpha_i\}.
$$
Methods for choosing $\{\alpha_i\}_{i=1}^m$ and for estimating $\{\sigma^2_i\}_{i=1}^m$, as well as experiments on simulated and real data, are given in~\citet{jiang2014unit}.

\subsection{Algorithms described in Section~\ref{sec:algorithm}}

\begin{algorithm}[H]
 \caption{A general algorithm for minimizing~(\ref{objective}).\label{alg:general}}
  \begin{algorithmic}
  \Procedure{algorithm.general}{$f_1,\ldots,f_n$}
  	\For{$i=1:n$}
  		\State Find region $C_i$ such that $f_i(\bmx)\leq0$ on $C_i$.
  	\EndFor
	\State Find all the pieces $\{A_j\}_{j=1}^{m}$ in the partition of $R^d$ formed by $\{C_i\}_{i=1}^n$.
	\State $s\leftarrow0$.
	\For{$j=1:m$}
		 \State Find the set of functions $\{f_k\}_{k\in I_j}$ that are not truncated on $A_j$.	
 		 \State $s\leftarrow\displaystyle\min\{s,\min_\bmx\sum_{k\in I_j}f_k(\bmx)\}$.
	\EndFor 
 	\State \textbf{return} $s$. 
  \EndProcedure
  \end{algorithmic}
\end{algorithm}

\begin{algorithm}[H]
 \caption{An algorithm for minimizing~(\ref{objective}) in 1-D.\label{alg:1d}}
  \begin{algorithmic}
  \Procedure{algorithm.1d}{$f_1,\ldots,f_n$}
  	\For{$i=1:n$}
  		\State Find the interval $C_i=[l_i, r_i]\subset R$ such that $f_i(x)\leq0$ on $C_i$.
  	\EndFor
  	\State Order all the $2n$ end-points of $\{C_i\}_{i=1}^n$ along the real line as $p_1<\cdots<p_{2n}$.
	\State $s\leftarrow0, I\leftarrow\emptyset$.
	\For{$j=1:2n$}
		\If{$p_j$ is the left end-point of an interval $C_k$}
			\State Add $k$ to set $I$.
		\ElsIf{$p_j$ is the right end-point of an interval $C_k$}
			\State Remove $k$ from set $I$.
		\EndIf
 		 \State $s\leftarrow\displaystyle\min\{s,\min_x\sum_{k\in I}f_k(x)\}$.
	\EndFor 
 	\State \textbf{return} $s$. 
  \EndProcedure
  \end{algorithmic}
\end{algorithm}

\begin{algorithm}[H]
 \caption{An algorithm for minimizing~(\ref{objective}) in 2-D.\label{alg:2d}}
  \begin{algorithmic}
  \Procedure{algorithm.2d}{$f_1,\ldots,f_n$}
  	\For{$i=1:n$}
  		\State Find $C_i\subset R^2$ such that $f_i(\bmx)\leq0$ on $C_i$.
  		\State Find $\partial C_i$, the boundary $C_i$.
  	\EndFor
	\State $s\leftarrow0$.
	\For{$i=1:n$}
		\State Find all the intersection points of $\partial C_i$ and $\partial C_k, k\neq i$.
	  	\State Sort all the intersection points along $\partial C_i$ clockwise as $\bmp_1,\ldots,\bmp_{n_i}$.
		\State Find a point $\bmp$ between $\bmp_1$ and $\bmp_{n_i}$ on $\partial C_i$.
		\State $I\leftarrow\{k:\bmp\in C_k\}, J\leftarrow I\setminus\{i\}$.
		\For{$j=1:n_i$}
			\If{$\bmp_j$ is the intersection point of $\partial C_i$ and $\partial C_k$ and $k\in I$}
				\State Remove $k$ from sets $I$ and $J$.
			\ElsIf{$\bmp_j$ is the intersection point of $\partial C_i$ and $\partial C_k$ and $k\not\in I$}
				\State Add $k$ to sets $I$ and $J$.
			\EndIf
 		\State $s\leftarrow\displaystyle\min\{s,\min_\bmx\sum_{k\in I}f_k(\bmx),\min_\bmx\sum_{k\in J}f_k(\bmx)\}$.
		\EndFor
	\EndFor
 	\State \textbf{return} $s$. 
  \EndProcedure
  \end{algorithmic}
\end{algorithm}

\begin{algorithm}[H]
 \caption{A cyclic coordinate descent algorithm for minimizing~(\ref{objective}) in high-dimensional settings.\label{alg:iterative}}
  \begin{algorithmic}
  \Procedure{algorithm.high-d}{$f_1,\ldots,f_n$}
  	\State Initialize $\bmx$ as $\bmx_0$.
  	\While{true}
	  	\For{$j=1:d$}
  			\State Fix all $x_k, k\neq j$, minimize the objective function as a univariate function of $x_j$ using Algorithm~\ref{alg:1d}.
	  	\EndFor
	  	\If{the change in $\bmx$ since the last iteration is less than a given tolerance level}
		 	\State \textbf{return} $\bmx$.
		\EndIf
  	\EndWhile
  \EndProcedure
  \end{algorithmic}
\end{algorithm}

\subsection{The $\Theta$-IPOD algorithm for robust linear regression}
\label{sec:IPOD}

\begin{algorithm}[H]
 \caption{The $\Theta$-IPOD algorithm for robust linear regression, adapted from Algorithm~2 in~\citet{she2012outlier}.\label{alg:IPOD}}
  \begin{algorithmic}
  \Procedure{$\Theta$-IPOD}{$\bmX\in R^{n\times p}, \bmy\in R^n, \bmlambda>0\in R^n, \bmgamma^{(0)}\in R^p$, and threshold operator $\Theta(\cdot; \cdot)$ which is taken as the hard-threshold operator $\Theta_h(\cdot; \cdot)$ in our paper}
	\State $\bmgamma\leftarrow\bmgamma^{(0)},\bmH\leftarrow\bmX(\bmX^T\bmX)^{-1}\bmX^T,\bmr\leftarrow\bmy-\bmH\bmy$.
  	\While{true}
  	 	\State $\bmgamma\leftarrow\Theta_h(\bmH\bmgamma+\bmr; \sqrt{\bmlambda})$.
	  	\If{the change in $\bmgamma$ since the last iteration is less than a given tolerance level}
		 	\State \textbf{return} $\hat\bmgamma\leftarrow\bmgamma$ and $\hat\bmbeta\leftarrow(\bmX^T\bmX)^{-1}\bmX^T(\bmy-\hat{\bmgamma})$.
		\EndIf
  	\EndWhile
  \EndProcedure
  \end{algorithmic}
\end{algorithm}

\subsection{The difference of convex (DC) functions algorithm}
\label{sec:DC}

Following~\citet{an1997solving}, we rewrite our objective function for sum of truncated quadratic functions
$$f(\bmx)=\sum_{i=1}^n\min\left(\frac12\bmx^T\bmA_i\bmx+\bmx^T\bmb_i+c_i, \lambda_i\right)$$
as $f(\bmx)=f_1(\bmx)-f_2(\bmx)$, where
$$f_1(\bmx)=\sum_{i=1}^n\frac12\bmx^T\bmA_i\bmx+\bmx^T\bmb_i+c_i$$
is a quadratic function, and
$$f_2(\bmx)=\sum_{i=1}^n\left(\frac12\bmx^T\bmA_i\bmx+\bmx^T\bmb_i+c_i-\lambda_i\right)_+.$$
Then, the DC algorithm iteratively minimizes a convex majorization of $f(\cdot)$ by replacing $f_2(\cdot)$ with its linear approximation at $\bmx^k$, until converge. That is,
$$\bmx^{k+1}=\argmin_\bmx\left\{f_1(\bmx)-[\nabla f_2(\bmx^k)]^T(\bmx-\bmx^k)\right\},$$
where $\nabla f_2(\bmx^k)$ is the gradient of $f_2(\bmx)$ evaluated at $\bmx^k$, and we have
$$\nabla f_2(\bmx^k)=\sum_{i=1}^n1\left(\frac12\bmx^T\bmA_i\bmx+\bmx^T\bmb_i+c_i>\lambda_i\right)(\bmA_i\bmx^k+\bmb_i).$$
Therefore,
$$\bmx^{k+1}=\argmin_\bmx\left\{\sum_{i=1}^n\frac12\bmx^T\bmA_i\bmx+\bmx^T\bmb_i+c_i-\sum_{i=1}^n1\left(\frac12\bmx^T\bmA_i\bmx+\bmx^T\bmb_i+c_i>\lambda_i\right)(\bmA_i\bmx^k+\bmb_i)^T(\bmx-\bmx^k)\right\}$$
for which we only need to minimize a quadratic function, and the solution exists in closed-form.

\subsection{The iterative marginal optimization (IMO) algorithm for signal and image restoration}
\label{sec:IMO}

Following~\citet{portilla2015efficient}, we rewrite our objective function
$$f(\bmx)=\sum_{i=1}^d(x_i-y_i)^2+w\sum_{i=1}^n\min\{x_i-x_{i+1})^2, \lambda\}.$$
as
$$f(\bmx)=||\bmH\bmx-\bmy||^2+w\sum_{i=1}^n\min\{(\bmphi_i^T\bmx)^2, \lambda\},$$
where $n=d-1,\bmH=\bmI_n$ is an identity matrix, $\bmPhi=(\bmphi_1,\ldots,\bmphi_n)^T\in R^{n\times d}$ with $\phi_{i,i}=-1,\phi_{i,i+1}=1$ and otherwise $\phi_{i,j}=0$ for all $i$ and $j$. We then minimize $f(\bmx)$ using the following iterative algorithm proposed in~\citet{portilla2015efficient}, where $\Theta_h(\cdot;\cdot)$ is the hard-threshold operator.

\begin{algorithm}[H]
 \caption{The iterative marginal optimization (IMO) algorithm for signal and image restoration, Adapted from Algorithm~1 in~\citet{portilla2015efficient}.\label{alg:IMO}}
  \begin{algorithmic}
  \Procedure{Threshold}{$\bmy\in R^{d}, \bmPhi\in R^{n\times d}, \bmlambda>0\in R^n$}
  	\State $\bmx\leftarrow\bmy$.
  	\While{true}
  		\State $\bmb\leftarrow\bmPhi\bmx$.
  		\State $\bma\leftarrow\Theta_h(\bmb;\sqrt{\bmlambda})$.
  		\State $\bmz\leftarrow w\bmPhi^T\bma$.
  		\State $\bmx\leftarrow(\bmH^T\bmH+w\bmPhi^T\bmPhi)^{-1}(\bmH^T\bmy+\bmz)$.
	  	\If{the change in $\bmx$ since the last iteration is less than a given tolerance level}
		 	\State \textbf{return} $\bmx$.
		\EndIf
  	\EndWhile
  \EndProcedure
  \end{algorithmic}
\end{algorithm}

\section{Proofs}

\begin{proof}[Proof of Proposition~\ref{equivalence}]~\\

\noindent To minimize~(\ref{likelihood}),
$$f(\bmbeta,\bmgamma)=\sum_{i=1}^n(y_i-\bmx^T_i\bmbeta-\gamma_i)^2+\lambda\sum_{i=1}^n1(\gamma_i\neq0),$$ 
notice that the minimization with respect to $\bmgamma$ can be performed componentwise. For each $\gamma_i$, if $\gamma_i=0$, we have
\beq
\label{choice1}
f(\bmbeta, \gamma_1,\ldots,\gamma_i=0,\ldots,\gamma_n)=\sum_{j \neq i}\left\{(y_j-\bmx^T_i\bmbeta-\gamma_j)^2+\lambda1(\gamma_j\neq0)\right\}+(y_i-\bmx^T_i\bmbeta)^2.
\eeq
On the other hand, if $\gamma_i\neq0$, we have
$$f(\bmbeta, \gamma_1,\ldots,\gamma_i\neq0,\ldots,\gamma_n)=\sum_{j \neq i}\left\{(y_j-\bmx^T_i\bmbeta-\gamma_j)^2+\lambda1(\gamma_j\neq0)\right\}+(y_i-\bmx^T_i\bmbeta-\gamma_i)^2+\lambda,$$
which is minimized at $\gamma_i=y_i-\bmx^T_i\bmbeta$, that is,
\beq
\label{choice2}
f(\bmbeta, \gamma_1,\ldots,\gamma_i=y_i-\bmx^T_i\bmbeta,\ldots,\gamma_n)=\sum_{j \neq i}\left\{(y_j-\bmx^T_i\bmbeta-\gamma_j)^2+\lambda1(\gamma_j\neq0)\right\}+\lambda.
\eeq
Comparing~(\ref{choice1}) with~(\ref{choice2}), it is easy to see that we should choose $\gamma_i=0$ if $(y_i-\bmx^T_i\bmbeta)^2<\lambda$ and $\gamma_i=y_i-\bmx^T_i\bmbeta$ othersise. Plugging the value of $\gamma_i$ into~(\ref{likelihood}), we have
$$f(\bmbeta, \bmgamma)=\sum_{i=1}^n\left[(y_i-\bmx^T_i\bmbeta)^2 1\{(y_i-\bmx^T_i\bmbeta)^2<\lambda\}+\lambda1\{(y_i-\bmx^T_i\bmbeta)^2\geq\lambda\}\right]=\sum_{i=1}^n \min \{(y_i-\bmx^T_i\bmbeta)^2,\lambda\}$$
which is the objective function $g(\bmbeta)$ in Proposition~\ref{equivalence}.
\end{proof}

\begin{proof}[Proof of Proposition~\ref{glm.equivalence}]~\\

\noindent Similar to the proof of Proposition~\ref{equivalence}, for the objective function in~(\ref{glm}), if $\gamma_i=0$, the $i$-th summand becomes
$
b(\bmx_i^T\bmbeta)-(\bmx_i^T\bmbeta)y_i.
$
Otherwise, if $\gamma_i \neq 0$, the $i$-th summand becomes 
$b(\bmx_i^T\bmbeta+\gamma_i)-(\bmx_i^T\bmbeta+\gamma_i)y_i+\lambda,$
which is minimized when
$$
y_i=\frac{\partial b(\bmx_i^T\bmbeta+\gamma_i)}{\partial\gamma_i}=g^{-1}(\bmx_i^T\bmbeta+\gamma_i)\Rightarrow \bmx_i^T\bmbeta+\gamma_i=g(y_i)
$$
which makes the $i$-th summand become
$\lambda^*:=b(g(y_i))- g(y_i)y_i+\lambda.$
The objective function can then be rewritten as:
$$\sum^n_{i=1} \min\{ b(\bmx_i^T\bmbeta)-(\bmx_i^T\bmbeta)y_i, \lambda^*\}$$
which completes the proof.
\end{proof}

\begin{proof}[Proof of Proposition~\ref{3sat}]~\\

\noindent Let $b_1,\ldots,b_n$ be $n$ Boolean variables, i.e., each $b_k$ only takes one of two possible values: TRUE or FALSE. For a 3-SAT problem $P$, suppose its formula is 
$$f(b_1,\ldots,b_n)=c_1\wedge\cdots\wedge c_m,$$
where $\wedge$ is the logical OR operator, and $\{c_i\}_{i=1}^m$ are the clauses\footnote{A clause is a disjunction of literals or a single literal. In a 3-SAT problem each clause has exactly three literals.} of P with 
$$c_i=(l_{i1}\vee l_{i2}\vee l_{i3}),$$ 
where $\vee$ is the logical AND operator, and $\{l_{ij}\}_{i=1}^m, j\in\{1,2,3\}$, are literals of~$P$. Each literal $l_{ij}$ is either a variable $b_k$ for which $l_{ij}$ is called a positive literal, or the negation of a variable $\neg b_k$ for which $l_{ij}$ is called a negative literal. Without loss of generality, suppose that each clause consists of exactly three literals, and that the three literals in each clause correspond to three distinct variables. The 3-SAT problem $P$ concerns about the satisfiability of $f(b_1,\ldots,b_n)$, i.e., whether there exists a possible assignment of values of $b_1,\ldots,b_n$ such that $f(b_1,\ldots,b_n)=$~TRUE.

We reduce the 3-SAT problem $P$ to the minimization of a sum of truncated convex functions $g(\bmx):R^n\rightarrow R$ as follows. Let $\bmx=(x_1,\ldots,x_n)\in R^n$ with each $x_k$ corresponds to a $b_k$ such that $b_k=\mbox{TRUE}$ if and only if $x_k>0$. For each clause $c_i=(l_{i1}\vee l_{i2}\vee l_{i3})$ of $P$, define a sum of seven truncated convex functions 
$$g_i(\bmx)=\sum_{t=1}^7\min(g_{it}(\bmx), 1)$$
where 
$$g_{it}(\bmx)=
\left\{\begin{array}{ll}
0 & \mbox{ if } \bmx\in S_{it1}\cap S_{it2}\cap S_{it3}\\
\infty & \mbox { otherwise}
\end{array}\right.$$
where $S_{itj}$ is one of the two half-spaces defined by $x_k>0$ and $x_k\leq0$, respectively, where $x_k$ is the variable corresponding to $l_{ij}$, that is, $l_{ij}=b_k$ or $l_{ij}=\neg b_k$. We choose $S_{itj}$ as the half-space defined by $x_k>0$ if and only if $(b(j,t)-\frac12)$ has the same sign as $l_{ij}$, where $b(j,t)$ is the $j$-th digit (from left to right) of $t$ when $t\in\{1,\ldots,7\}$ is represented as three binary digits. For instance, for a clause $c_i=(b_1\vee\neg b_2\vee b_3)$, we have
$$g_{i1}(\bmx)=
\left\{\begin{array}{ll}
0 & \mbox{ if } x_1\leq0,x_2>0,x_3>0\\
\infty & \mbox { otherwise}
\end{array}\right.$$
and 
$$g_{i7}(\bmx)=
\left\{\begin{array}{ll}
0 & \mbox{ if } x_1>0,x_2\leq0,x_3>0\\
\infty & \mbox { otherwise}
\end{array}\right.$$
Since all the half-spaces, as well as their intersections, are convex sets, all the $g_{it}(\bmx)$'s are convex functions. Furthermore, since the regions in which $g_{it}(\bmx)=0, t\in\{1,\ldots,7\}$, are disjoint, it is easy to verify that $g_i(\bmx)$ can only take one of two possible values
$$g_i(\bmx)=
\left\{\begin{array}{ll}
6 & \mbox{ if } c_i \mbox{ is satisfied by the assigned values of } b_1,\ldots,b_n \\
7 & \mbox { otherwise}
\end{array}\right.$$
where we choose $b_k=\mbox{TRUE}$ if and only if $x_k>0$. The reduction is then completed by noticing that the 3-SAT problem $P$ is satisfiable if and only if the minimum value of the function $g(\bmx)=\sum_{i=1}^mg_i(\bmx)$ is~$6m$, and that it is easy to see that the reduction can be done in polynomial time.

\end{proof}

\begin{proof}[Proof of Proposition~\ref{ignore}]~\\

\noindent On one hand, we have
\beq
\label{forward}
\min_\bmx\sum_{i=1}^{n}\min\{f_i(\bmx),0\}=\min_j\min_{\bmx\in A_j} \sum_{k \in I_j} f_k(\bmx)\geq\min_j\min_{\bmx} \sum_{k \in I_j} f_k(\bmx),
\eeq
On the other hand, we have
\beq
\label{backward}
\begin{array}{ll}
\displaystyle\min_j\min_\bmx \sum_{k \in I_j} f_k(\bmx)&\displaystyle\geq\min_j\min_{\bmx} \sum_{k \in I_j}\min\{f_k(\bmx),0\}\\
&\displaystyle\geq\min_j\min_\bmx\sum_{i=1}^n\min\{f_i(\bmx),0\}\\
&\displaystyle=\min_\bmx\sum_{i=1}^{n}\min\{f_i(\bmx),0\}
\end{array}
\eeq
Putting~(\ref{forward}) and~(\ref{backward}) together, we have
$$\min_\bmx\sum_{i=1}^{n}\min\{f_i(\bmx),0\}=\min_j\min_\bmx \sum_{k \in I_j} f_k(\bmx).$$
\end{proof}

\section{Supplementary figures and tables}

\begin{figure}[H]
\begin{center}
\includegraphics[width=1\textwidth]{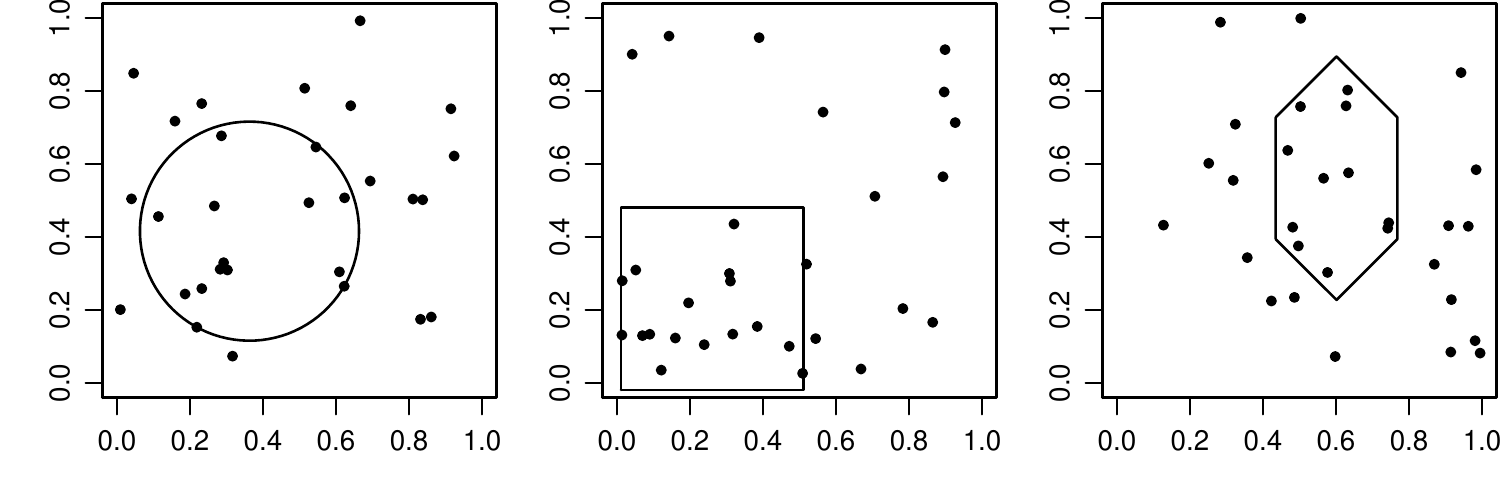}
\end{center}
\caption{Placement of different convex shapes to cover the maximum number of points uniformly sampled from the unit square. \label{fig:shape}}
\end{figure}

\begin{figure}[H]
\begin{center}
\includegraphics[width=1\textwidth]{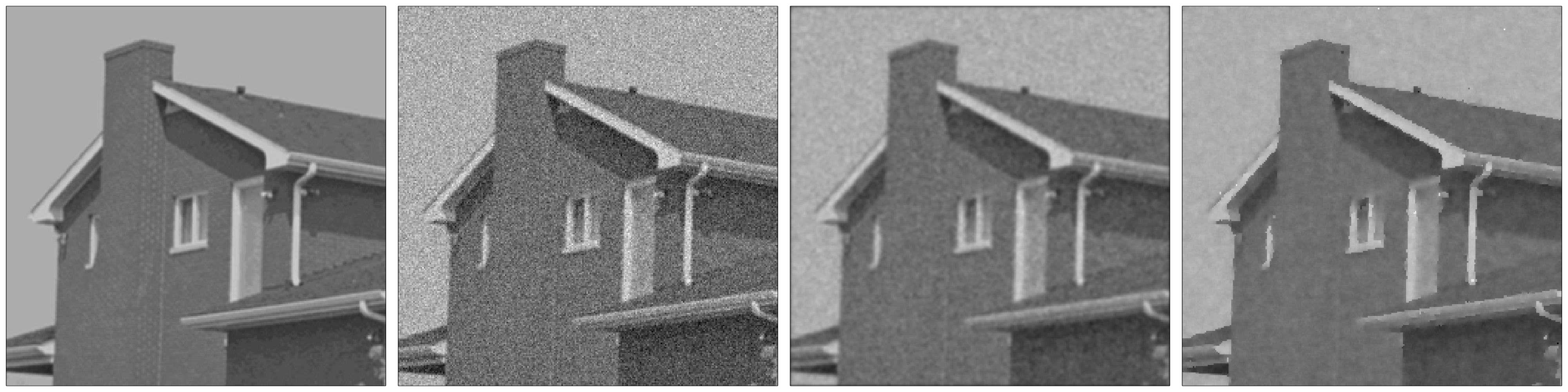}
\includegraphics[width=1\textwidth]{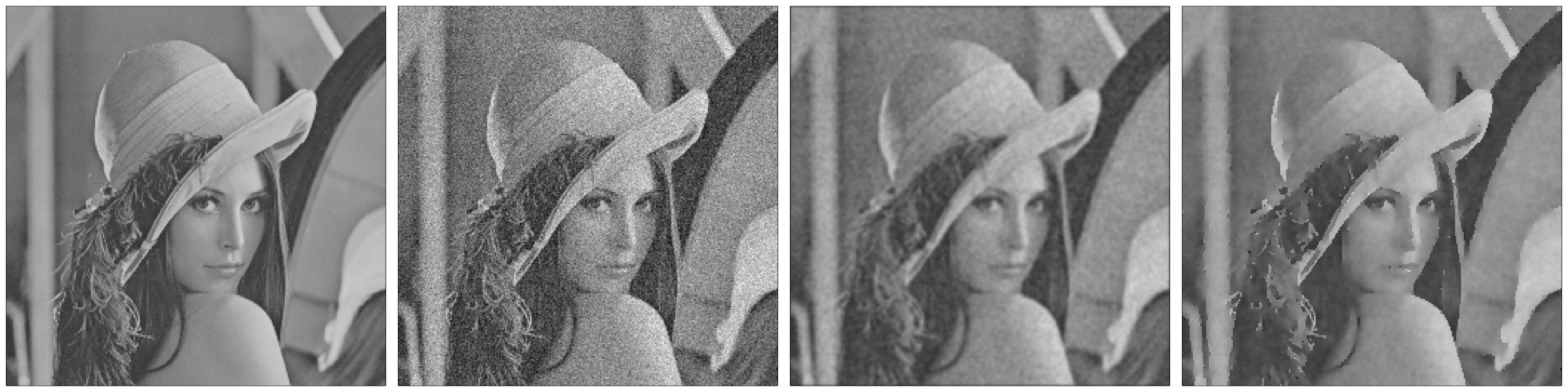}
\end{center}
\caption{Restoration of images. For each row, from left to right: original image, image with Gaussian noise added, image restored using Gaussian smoothing with a $5\times5$ kernel and image restored using proposed algorithm. \label{fig:image.more}}
\end{figure}

\begin{table}[H]
\caption{Comparison of different methods for outlier detection in simple linear regression. The table shows the leverages of outliers (L), percents of outliers (O\%) and mean percents of masking and swamping for all the methods using 100 simulated replicates. The standard errors of the means are given in parentheses. \label{tab:outlier}}
\centering
Masking\\
\begin{tabular}{ccccccc}
  \hline
L & O\% & MM & LTS & GY & IPOD & Proposed \\ 
  \hline
0 & 5 & 0.6 (0.3) & 0.6 (0.3) & 1.0 (0.4) & 0.6 (0.3) & 0.8 (0.4) \\ 
  0 & 10 & 2.3 (0.5) & 2.1 (0.5) & 2.4 (0.5) & 2.0 (0.5) & 2.0 (0.5) \\ 
  0 & 20 & 1.8 (0.3) & 1.8 (0.3) & 1.9 (0.3) & 1.5 (0.3) & 1.4 (0.3) \\ 
  0 & 30 & 3.1 (0.3) & 2.6 (0.3) & 3.2 (0.3) & 2.3 (0.3) & 2.3 (0.3) \\ 
  0 & 45 & 11.1 (0.5) &  4.8 (0.3) &  9.8 (0.4) &  2.5 (0.3) &  2.5 (0.3) \\ 
  0 & 60 & 33.2 (0.5) & 24.3 (0.5) & 33.0 (0.5) & 34.4 (0.7) &  2.8 (0.4) \\ 
  20 & 5 & 2.0 (0.6) & 2.4 (0.7) & 2.0 (0.6) & 1.8 (0.6) & 1.8 (0.6) \\ 
  20 & 10 & 3.1 (0.6) & 2.8 (0.6) & 3.1 (0.6) & 2.7 (0.5) & 2.8 (0.5) \\ 
  20 & 20 & 4.8 (0.7) & 2.8 (0.4) & 3.7 (0.5) & 3.2 (0.5) & 2.8 (0.4) \\ 
  20 & 30 & 14.1 (1.0) &  5.7 (0.6) &  8.9 (0.7) &  5.5 (0.7) &  3.9 (0.4) \\ 
  20 & 45 & 34.5 (0.7) & 13.8 (0.8) & 29.6 (0.9) & 18.0 (1.4) &  7.5 (0.9) \\ 
  20 & 60 & 34.0 (0.7) & 36.1 (0.5) & 34.9 (0.6) & 35.0 (0.5) & 24.1 (1.3) \\ 
   \hline
\end{tabular}
\vspace{0.5cm}\\
Swamping\\
\begin{tabular}{ccccccc}
  \hline
L & O\% & MM & LTS & GY & IPOD & Proposed \\ 
  \hline
0 & 5 & 1.0 (0.1) & 1.4 (0.1) & 1.0 (0.1) & 1.1 (0.1) & 1.0 (0.1) \\ 
  0 & 10 & 1.1 (0.1) & 1.2 (0.1) & 1.1 (0.1) & 1.2 (0.1) & 1.2 (0.1) \\ 
  0 & 20 & 1.3 (0.1) & 1.3 (0.1) & 1.3 (0.1) & 1.4 (0.1) & 1.3 (0.1) \\ 
  0 & 30 & 1.3 (0.1) & 1.3 (0.1) & 1.5 (0.2) & 1.4 (0.1) & 1.3 (0.2) \\ 
  0 & 45 & 19.0 (1.1) &  2.9 (0.3) & 15.1 (0.8) &  1.4 (0.1) &  1.3 (0.1) \\ 
  0 & 60 & 96.1 (0.6) & 73.2 (1.4) & 95.5 (0.6) & 90.5 (1.8) &  2.0 (0.6) \\ 
  20 & 5 & 1.5 (0.1) & 1.8 (0.1) & 1.4 (0.1) & 1.5 (0.1) & 1.5 (0.1) \\ 
  20 & 10 & 1.1 (0.1) & 1.3 (0.1) & 1.1 (0.1) & 1.2 (0.1) & 1.3 (0.1) \\ 
  20 & 20 & 1.2 (0.1) & 1.1 (0.1) & 1.1 (0.1) & 1.1 (0.1) & 1.1 (0.1) \\ 
  20 & 30 & 4.0 (0.4) & 1.6 (0.2) & 1.9 (0.2) & 1.6 (0.1) & 1.4 (0.1) \\ 
  20 & 45 & 30.2 (0.8) &  4.5 (0.3) & 20.7 (0.7) &  8.6 (0.9) &  2.4 (0.3) \\ 
  20 & 60 & 44.9 (0.9) & 34.8 (0.9) & 42.8 (0.9) & 33.0 (1.3) & 13.5 (1.1) \\ 
   \hline
\end{tabular}
\end{table}

\begin{table}[H]
\caption{Stopping criteria of the simulation studies in Sections~\ref{sec:quadratic.experiment} and ~\ref{sec:image.experiment}\label{tab:settings}}
\centering
\begin{tabular}{@{}ccccc@{}}
\hline
            & \begin{tabular}[c]{@{}c@{}}Maximal number\\  of function \\ evaluations\end{tabular} & \begin{tabular}[c]{@{}c@{}}Maximal number\\  of iterations\end{tabular} & Tolerance & \begin{tabular}[c]{@{}c@{}}Maximal steps when\\ no improvement\\  in the estimate\end{tabular} \\ 
            \hline
DIRECT      & $10^4$                                                                               & -                                                                       & $10^{-8}$ & -                                                                                              \\
StoGO      & $10^4$                                                                               & -                                                                       & $10^{-8}$ & -                                                                                              \\
SA          & $10^4$                                                                               & $10^4$  & -         & $10^6$                                                                                         \\
PSO         & $10^4$                                                                               & $10^4$                                                                    & $10^{-8}$ & -                                                                                              \\ 
IMO/DC         &                                                                                      & $10^4$                                                                    & $10^{-8}$ & -                                                                                              \\ 
Proposed (high-D)  &                                                                                      & $10^4$                                                                    & $10^{-8}$ & -                                                                                              \\ 
\hline
\end{tabular}
\end{table}

\begin{table}[H]
\caption{Comparison of different algorithms for global optimization of the sum of $50$ randomly generated truncated quardratic funstions in 2-D. The table shows the complexities of the functions ($C$) as well as mean success rates (in percents) and running times (in seconds) for all the methods using 100 simulated replicates. The standard errors of the means are given in parentheses. \label{tab:quadratic.2D}}
\centering
Success Rate\\
\begin{tabular}{ccccccc}
  \hline
C & DIRECT & StoGO & SA & PSO & DC & Proposed \\ 
  \hline
1 & 100.0 (0.0) &  86.0 (3.5) &  98.0 (1.4) &  99.0 (1.0) &  27.0 (4.5) & 100.0 (0.0) \\ 
  5 &  99.0 (1.0) &  74.0 (4.4) &  97.0 (1.7) &  93.0 (2.6) &   9.0 (2.9) & 100.0 (0.0) \\ 
  10 &  88.0 (3.3) &  57.0 (5.0) &  85.0 (3.6) &  72.0 (4.5) &   1.0 (1.0) & 100.0 (0.0) \\ 
   \hline
\end{tabular}
\vspace{0.5cm}\\
Running Time\\
\begin{tabular}{ccccccc}
  \hline
C & DIRECT & StoGO & SA & PSO & DC & Proposed \\ 
  \hline
1 & 0.45 (0.01) & 2.76 (0.02) & 0.40 (0.00) & 2.29 (0.05) & 0.50 (0.03) & 3.07 (0.03) \\ 
  5 & 0.42 (0.01) & 2.62 (0.07) & 0.39 (0.00) & 2.66 (0.08) & 2.95 (0.29) & 2.62 (0.03) \\ 
  10 & 0.44 (0.05) & 2.41 (0.02) & 0.37 (0.00) & 2.82 (0.12) & 8.04 (0.71) & 2.35 (0.03) \\ 
   \hline
\end{tabular}
\end{table}

\end{document}